\documentclass[a4paper,reqno,twoside]{amsart}
\usepackage{bbm}
\usepackage[left=3.5cm,right=3.5cm,top=3.5cm,bottom=3.5cm,footskip=1.5cm,headsep=1cm,bindingoffset=0.cm,]{geometry}

\usepackage[utf8]{inputenc}
\usepackage[english]{babel}
\usepackage[T1]{fontenc} % To switch to the T1 encoding
\usepackage{lmodern} % To switch to Latin Modern
\rmfamily % To load Latin Modern Roman and enable the following NFSS declarations.
\DeclareFontShape{T1}{lmr}{b}{sc}{<->ssub*cmr/bx/sc}{}
\DeclareFontShape{T1}{lmr}{bx}{sc}{<->ssub*cmr/bx/sc}{}

\numberwithin{equation}{section}
\usepackage{amsmath}
\usepackage{amssymb}
\usepackage{amsthm}
\usepackage{amsfonts}
\usepackage{mathtools}

\usepackage{mathdots}
\usepackage{caption}
\usepackage{subcaption}

\usepackage{dsfont} % for 1 bbmath
\usepackage[format=hang]{caption}
\usepackage[normalem]{ulem}

\usepackage{standalone}
\usepackage{graphicx}
\usepackage{imakeidx}
\makeindex
\usepackage{setspace}
\usepackage{appendix}
\usepackage{pdfpages}
\usepackage{upgreek}
\usepackage{tabulary}
\usepackage{booktabs}
\usepackage{extarrows}
\usepackage{tabularx}
\usepackage{float}
\restylefloat{table}
\usepackage{enumitem}
\usepackage{csquotes}
\usepackage{bm}
\usepackage{orcidlink}
\usepackage{lipsum}

\usepackage{tikz}
\usepackage{tikz-layers}
\usepackage{tikz-cd}
\usetikzlibrary{matrix,positioning,decorations.pathreplacing,calc}
\usetikzlibrary{
    decorations.text,%
    decorations.markings,%
    shadows}
\usetikzlibrary{calc,intersections}
\usepackage{adjustbox}
\usepackage{graphicx}

\usepackage[
    style=numeric,
    sorting=nty,
    maxnames=99,
    maxalphanames=5,
    natbib=true,
    sortcites]{biblatex}

\DeclareNameAlias{default}{family-given}

\AtEveryBibitem{% Clean up the bibtex rather than editing it
    \clearfield{url}
    \clearfield{issn}
    \clearfield{isbn}
    \clearfield{urldate}

    \ifentrytype{book}{
        \clearfield{pages}}{% 
    }
}

\usepackage{xargs}                      % Use more than one optional parameter in a new commands
\usepackage[prependcaption,textsize=tiny,textwidth=3cm]{todonotes}
\newcommandx{\unsure}[2][1=]{\todo[linecolor=red,backgroundcolor=red!25,bordercolor=red,#1]{#2}}
\newcommandx{\change}[2][1=]{\todo[linecolor=blue,backgroundcolor=blue!25,bordercolor=blue,#1]{#2}}
\newcommandx{\info}[2][1=]{\todo[linecolor=OliveGreen,backgroundcolor=OliveGreen!25,bordercolor=OliveGreen,#1]{#2}}
\newcommandx{\improvement}[2][1=]{\todo[linecolor=black,backgroundcolor=black!25,bordercolor=black,#1]{#2}}
\newcommandx{\thiswillnotshow}[2][1=]{\todo[disable,#1]{#2}}

\usepackage[noabbrev,capitalize]{cleveref}
\crefname{proposition}{Proposition}{Propositions}
\crefname{equation}{}{}

\newtheorem{theorem}{Theorem}[section]
\newtheorem{lemma}[theorem]{Lemma}
\newtheorem{proposition}[theorem]{Proposition}
\newtheorem{corollary}[theorem]{Corollary}

\theoremstyle{definition}
\newtheorem{definition}[theorem]{Definition}

\newtheorem{remark}[theorem]{Remark}

\crefname{assumption}{Assumption}{Assumptions}
\crefname{definition}{Definition}{Definitions}
\crefname{corollary}{Corollary}{Corollaries}
\crefname{enumi}{item}{items}

\creflabelformat{subfigure}{#2\textsc{#1}#3}

\usepackage[final]{microtype}

\makeatletter
\newsavebox\myboxA
\newsavebox\myboxB
\newlength\mylenA

\newcommand*\xoverline[2][0.75]{%
  \sbox{\myboxA}{$\m@th#2$}%
  \setbox\myboxB\null% Phantom box
  \ht\myboxB=\ht\myboxA%
  \dp\myboxB=\dp\myboxA%
  \wd\myboxB=#1\wd\myboxA% Scale phantom
  \sbox\myboxB{$\m@th\overline{\copy\myboxB}$}%  Overlined phantom
  \setlength\mylenA{\the\wd\myboxA}%   calc width diff
  \addtolength\mylenA{-\the\wd\myboxB}%
  \ifdim\wd\myboxB<\wd\myboxA%
    \rlap{\hskip 0.5\mylenA\usebox\myboxB}{\usebox\myboxA}%
  \else
    \hskip -0.5\mylenA\rlap{\usebox\myboxA}{\hskip 0.5\mylenA\usebox\myboxB}%
  \fi}
\makeatother

\usepackage{fancyhdr}

\fancypagestyle{plain}{%
    \fancyhead[C]{}
    \fancyfoot[C]{}
    
    }
\fancypagestyle{nosection}{%
    \fancyhead[CE]{}
    \fancyhead[CO]{}
    \fancyfoot[CE,CO]{\thepage}
}

\usepackage{tikz}

\usetikzlibrary{matrix} 
\usetikzlibrary{arrows,arrows.meta,bending} %new code

\usepackage{nicematrix}

\usetikzlibrary{hobby}

\usepackage{calc} 

\newcommand{\lz}{\ell^2(\Z)}

\newcommand{\bseq}[2]{{\bm{#1}^{(#2)}}}

\DeclareMathOperator{\Z}{\mathbb{Z}}

\DeclareMathOperator{\R}{\mathbb{R}}
\DeclareMathOperator{\C}{\mathbb{C}}

\renewcommand{\tilde}{\widetilde}

\renewcommand{\qedsymbol}{$\blacksquare$}

\DeclareMathOperator{\diag}{diag}

\renewcommand{\epsilon}{\varepsilon}

\let\emptyset\varnothing

\renewcommand{\tilde}{\widetilde}

\usepackage{tcolorbox}
\usetikzlibrary{tikzmark,calc}

\newcommand{\mc}[1]{\mathcal{#1}}
\newcommand{\on}[1]{\operatorname{#1}}

\newcommand{\abs}[1]{\left\lvert#1\right\rvert}

\newcommand{\norm}[1]{\left\lVert#1\right\rVert}

\newcommandx{\silvio}[2][1=]{\todo[linecolor=blue,backgroundcolor=blue!25,bordercolor=blue,#1]{Silvio: #2}}

\newcommandx{\bowen}[2][1=]{\todo[linecolor=blue,backgroundcolor=blue!25,bordercolor=blue,#1]{bowen: #2}}

\newcommandx{\alex}[2][1=]{\todo[linecolor=red,backgroundcolor=red!25,bordercolor=red,#1]{Alex: #2}}

\newcommandx{\jiayu}[2][1=]{\todo[linecolor=red,backgroundcolor=red!25,bordercolor=red,#1]{Jiayu: #2}}

\title[Long-range interactions and Anderson localisation]{Long-range interactions and Anderson localisation for one-dimensional high-contrast resonator chain}

\begin{document}
  \author[H. Ammari]{Habib Ammari \,\orcidlink{0000-0001-7278-4877}}
  \address{\parbox{\linewidth}{Habib Ammari\\
   ETH Z\"urich, Department of Mathematics, Rämistrasse 101, 8092 Z\"urich, Switzerland, \href{http://orcid.org/0000-0001-7278-4877}{orcid.org/0000-0001-7278-4877}}.}
   \email{habib.ammari@math.ethz.ch}

\author[S. Barandun]{Silvio Barandun\,\orcidlink{0000-0003-1499-4352}}
  \address{\parbox{\linewidth}{Silvio Barandun\\
 Massachusetts Institute of Technology, Department of Mathematics, United States of America, \href{http://orcid.org/0000-0003-1499-4352}{orcid.org/0000-0003-1499-4352}}.}
 \email{barandun@math.mit.edu}

 \author[J. Qiu]{Jiayu Qiu} 
 \address{\parbox{\linewidth}{Jiayu Qiu\\
 Department of Mathematics, ETH Z\"{u}rich, R\"{a}mistrasse 101, CH-8092 Z\"{u}rich, Switzerland}}
 \email{jiayu.qiu@sam.math.ethz.ch}

  \author[A. Uhlmann]{Alexander Uhlmann\,\orcidlink{0009-0002-0426-6407}}
    \address{\parbox{\linewidth}{Alexander Uhlmann\\
    ETH Z\"urich, Department of Mathematics, Rämistrasse 101, 8092 Z\"urich, Switzerland, \href{http://orcid.org/0009-0002-0426-6407}{orcid.org/0009-0002-0426-6407}}.}
   \email{alexander.uhlmann@sam.math.ethz.ch}

\begin{abstract}
Spectral and transport properties of high-contrast resonator systems can be described in the subwavelength regime in terms of the so-called capacitance operator. In this paper, we consider an infinitely periodic chain of high-contrast resonators in three dimensions. The first result is a precise estimate of the off-diagonal decay rate of the capacitance operator $\mathcal{C}$. Importantly, we demonstrate that the decay rate is long-range and critical: as $|n-m|\to\infty$,
\begin{equation*}
\mathcal{C}(n,m)\sim \frac{1}{|n-m|\log^2|n-m|},
\end{equation*}
which is $\ell^1$ summable but slower than quadratic. This borderline decay of the off-diagonal entries makes the present proof of Anderson localisation with arbitrary disorder, which is observed numerically in this paper, out of reach; we hope that this physical example of classical wave systems with critical long-range interactions provides new insight in the field of Anderson localisation. As the second main result, based on the off-diagonal decay estimate, we prove a strong convergence of the finite capacitance operator, which corresponds to a truncated chain, to the capacitance operator as the size of the truncated chain grows to infinity. Using this strong convergence, we improve the results of [Ammari et al., SIAM J. Math. Anal., 2023 and Bull. London Math Soc., 2025] by presenting a rigorous estimate of the convergence rate of the spectrum.
\end{abstract}

\maketitle

 \par

\bigskip


\section{Introduction and motivation}
While the study of periodic systems has for a long time dominated the physical literature from quantum mechanics to classical wave systems, mainly thanks to powerful machinery such as the Floquet-Bloch theory, an increasing focus has been dedicated to disordered systems, that is, systems where periodicity is broken since the seminal paper by Anderson \cite{anderson1958Absence}. A common feature for the disordered system is that, by adding suitable randomness, the spectrum becomes a pure point spectrum and only localised modes exist in the system, which is the so-called \textit{Anderson localisation}. Interestingly, one of the most interesting phenomena is Anderson localisation in disordered systems with long-range interactions, where the governing operator (e.g., Hamiltonian) exhibits a power-law off-diagonal decay rate. In contrast to their tight-binding counterparts, for which the off-diagonal entries of the governing operators decay exponentially, the long-range interactions impose great challenge to the mathematical study of Anderson localisation; see, e.g., \cite{billy.josse.ea2008Direct,carmona1982Exponential,figotin.klein1994Localization, shi-adv26, vainberg,longrange1,longrange2,shapiro,shi2021multiscale} for discussions.

Much of the research efforts on this topic focus on one-dimensional systems on the one hand because these are already challenging, but on the other hand also because the long-range interactions that arise in higher dimensions render classical techniques (i.e. transfer matrices) inapplicable. In this paper, we consider an infinite periodic chain of high-contrast resonators embedded in a three-dimensional space. The governing operator for this classical wave system in the subwavelength regime is the so-called \textit{(one-dimensional) capacitance operator}, denoted as $\mathcal{C}$. As the first main result of this paper, we prove that the capacitance operator $\mathcal{C}$ exhibits a critical long-range off-diagonal decay rate: as $|n-m|\to\infty$, it is estimated as
\begin{equation} \label{eq_intro_1}
\mathcal{C}(n,m)\sim \frac{1}{|n-m|\log^2|n-m|},
\end{equation}
which is $\ell^1$ summable but slower than quadratic. This borderline decay of the off-diagonal entries makes the present proof of Anderson localisation with arbitrary disorder out of reach: for one-dimensional systems, the Anderson localisation with arbitrary disorder is only proved for systems with off-diagonal decay rate faster than $1/|n-m|^4$ \cite{Molchanov1999localization_1d_long_range}. Nevertheless, we numerically demonstrate that the Anderson localisation still appears in our system under disorder with arbitrary strength. On the other hand, the right side of \eqref{eq_intro_1} is absolutely summable along the off-diagonal entries: by this result, one can prove that Anderson localisation occurs when the disorder is sufficiently large; see \cite{aizenman.molchanov1993Localization,shi2023almost_periodic_long_range}. We hope that this physical example of classical wave systems with critical long-range interactions provides new insight in the field of Anderson localisation.

In the second part of this paper, based on the off-diagonal decay estimate \eqref{eq_intro_1}, we prove a strong convergence of the finite capacitance operator, which corresponds to a truncated chain, to the capacitance operator as the size of the truncated chain grows to infinity. Using this strong convergence, we prove the convergence of the spectrum, which reproduces the results of \cite{ammari.davies.ea2025Spectral, ammari.davies.ea2023Convergence} with a much simpler proof. Moreover, we obtain a rigorous estimate of the convergence rate of the discrete part of the spectrum, which highly improves the former results.

The paper is structured as follows:

\begin{itemize}
    \item In Sections \ref{sec:setting} and \ref{sec:periodic}, we outline the mathematical set-up of the finite and infinite high-contrast resonator system, along with the definition of (finite and infinite) capacitance operators. The first main result of this paper, Theorem \ref{thm_off_diag_decay}, is presented in Section 3.2. The major step of the proof is to establish an estimate of the Floquet component of the capacitance operator; see Section 3.1 for details.   
    \item In Section \ref{sec:SC}, we establish the strong convergence of the finite capacitance operator corresponding to a truncated chain as the size of truncation grows to infinity. As a consequence, the convergence of the spectrum is proved in Section 4.2.
    \item In Section \ref{sec:anderson}, we numerically illustrate the Anderson localisation in the high-contrast resonator chain. We show that, under random perturbation of arbitrary strength, the spectrum becomes a pure point and only localised modes exist. This confirms physical intuition and provides new insights in the mathematical study of Anderson localisation.
\end{itemize}

\section{Setting and problem formulation for finite systems}\label{sec:setting}

As the study of the capacitance operator starts with finite structures \cite{cbms}, we will first introduce the definition of finite capacitance operators. Then, in the next section, we introduce its infinite counterpart.

We begin with a bounded connected domain $D\Subset Y:=\{x=(x_1,x_\perp)\in \mathbb{R}\times\mathbb{R}^2:\ 0\le x_1\le a\}$ with a boundary of class $C^{2,\eta}$ ($\eta\in(0,1)$), for some positive constant $a>0$. The bounded domain $D$ represents a single resonator embedded in the free space. Next, we consider the translation image of $D$ along the first direction with spacing $a$, i.e.,
\begin{equation*}
    D_{n}:=D+na\bm{e}_1 ,
\end{equation*}
with $\bm{e}_1=(1,0,0)^{\top}$. The structure of our interest is \textit{the finite and infinite chain of resonators aligned along the first direction}. Let us start with the finite array consisting of $N$ resonators, which is defined as
\[
    \mathcal{D}^{(N)} \coloneqq \bigcup_{i=1}^ND_i .
\]

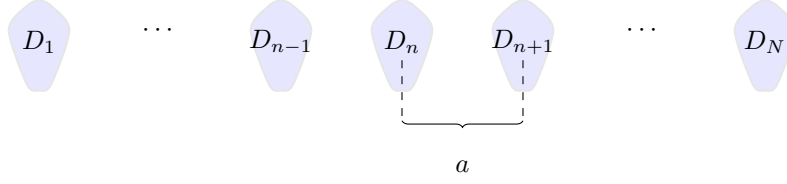
\begin{figure}
    \centering
    \begin{tikzpicture}[scale=0.8]
%resonators
\foreach \x in {-2,0,2} { 
        \draw [thick,fill=blue,opacity=0.1] plot [smooth,samples=400, tension=1] coordinates {({0+\x},{0.5}) ({0.2+\x},{0.4}) ({0.5+\x},{0}) ({0.2+\x},{-0.9}) ({0+\x},{-1}) ({-0.2+\x},{-0.9}) ({-0.5+\x},{0}) ({-0.2+\x},{0.4}) ({0+\x},{0.5})};
    }
\node[scale=1,thick] at (4,0) {$\cdots$};
\node[scale=1,thick] at (-4,0) {$\cdots$};
\draw [thick,fill=blue,opacity=0.1] plot [smooth,samples=400, tension=1] coordinates {({0+6},{0.5}) ({0.2+6},{0.4}) ({0.5+6},{0}) ({0.2+6},{-0.9}) ({0+6},{-1}) ({-0.2+6},{-0.9}) ({-0.5+6},{0}) ({-0.2+6},{0.4}) ({0+6},{0.5})};
\draw [thick,fill=blue,opacity=0.1] plot [smooth,samples=400, tension=1] coordinates {({0-6},{0.5}) ({0.2-6},{0.4}) ({0.5-6},{0}) ({0.2-6},{-0.9}) ({0-6},{-1}) ({-0.2-6},{-0.9}) ({-0.5-6},{0}) ({-0.2-6},{0.4}) ({0-6},{0.5})};
%label of resonators
\node[scale=1,thick] at ({-6},{-0.2}) {$D_1$};
\node[scale=1,thick] at ({-2},{-0.2}) {$D_{n-1}$};
\node[scale=1,thick] at ({0},{-0.2}) {$D_{n}$};
\node[scale=1,thick] at ({2},{-0.2}) {$D_{n+1}$};
\node[scale=1,thick] at ({6},{-0.2}) {$D_{N}$};
%spacing
\draw[dashed] ({0},{-0.5})--({0},{-1.5});
\draw[dashed] ({2},{-0.5})--({2},{-1.5});
\draw[decorate,decoration={brace,mirror}] ({0},{-1.5}) -- ({2},{-1.5});
\node[below,scale=1] at (1,-2) {$a$};
\end{tikzpicture}
    \caption{A finite chain of $N$ resonators.}
    \label{fig_N_resonator}
\end{figure}

The finite array $\mathcal{D}^{(N)}$ is sketched in Figure \ref{fig_N_resonator}. We assume that the material parameters are constant inside and outside of the resonators and denote by $\rho, \kappa, \rho_i, \kappa_i$ the material density and the bulk modulus outside and within the $i$\textsuperscript{th} resonator, respectively. The associated wave speeds are $v \coloneqq \sqrt{\kappa/\rho}$ and $v_i \coloneqq \sqrt{\kappa_i/\rho_i}$ yielding $k = \omega/v$ and $k_i = \omega/v_i$ for the wave numbers outside and inside the resonators. The resonance problem is then determined by the following coupled system of Helmholtz equations:
\begin{equation}\label{eq:hheq}
    \left\{\begin{array}{ll}
\Delta u + k^2 u = 0 & \text{in } \mathbb{R}^3 \setminus \overline{\mathcal{D}^{(N)}}, \\
\Delta u + k_i^2 u = 0 & \text{in } D_i, \\
u|_+ - u|_- = 0 & \text{on } \partial \mathcal{D}^{(N)}, \\
\delta_i \left.\dfrac{\partial u}{\partial \nu}\right|_+ - \left.\dfrac{\partial u}{\partial \nu}\right|_- = 0 & \text{on } \partial D_i\\
\multicolumn{2}{l}{\lim_{|\bm{x}| \to \infty} |\bm{x}| \left( \frac{\partial u}{\partial |\bm{x}|} - ik u \right) = 0.}
\end{array}\right.
\end{equation}
Here, $\delta_i = \rho_i / \rho$ is called the \emph{contrast ratio} of the $i$\textsuperscript{th} resonator. We interchangeably use $\partial/\partial \nu$ and $\partial_\nu$ to denote the normal derivative.

In the following, we shall assume that $\delta_i = \delta$ for all $i=1, \dots, N$. We are interested in the \emph{high-contrast limit} where $\delta\to 0$ while $v_i, v = \mc O(1)$. In this regime, we investigate the following solutions.
\begin{definition}[Subwavelength resonant frequency]
    Given $\delta>0$, $\omega(\delta)\in \C$ is called a \emph{subwavelength resonant frequency} if
    \begin{enumerate}
        \item[(i)] \cref{eq:hheq} admits a resonant solution with $\omega=\omega(\delta)$;
        \item[(ii)] $\omega(\delta)=\mathcal{O}(\sqrt{\delta})$.
    \end{enumerate}
    When it is clear from context, we will simply write $\omega$ for $\omega(\delta)$.
\end{definition}

To define the capacitance operator associated with this $N$-resonator system, we first recall the Green function of the free-space Laplacian equation:
\begin{equation}\label{eq:greensfn}
    G(\bm{x},\bm{y}) = -\frac{1}{4\pi \abs{\bm{x}-\bm{y}}}.
\end{equation}
Then we define the \emph{single-layer potential} on $\partial \mathcal{D}^{(N)}$
\begin{align}
    \mc S_{\mathcal{D}^{(N)}}:L^2(\partial \mathcal{D}^{(N)})&\to H^1(\partial \mathcal{D}^{(N)})\\
    \varphi &\mapsto \int_{\partial \mathcal{D}^{(N)}} G(\bm{x}-\bm{y})\varphi(\bm{y}) \mathrm{d\sigma}_{\bm{y}},
\end{align}
It is known that the operator $S_{\mathcal{D}^{(N)}}$ is invertible (see \cite[Lemma 2.2]{ammari.davies.ea2024Functional}). This allows us to define the capacitance operator as follows:
\begin{definition}[Finite capacitance operator]\label{def:capmat}
    Letting $\mathbbm{1}_{\partial D_j}$ be the indicator function on the $j$\textsuperscript{th} resonator, we define the characteristic solution
    \[
        \phi_j = (S_{\mathcal{D}^{(N)}})^{-1}[\mathbbm{1}_{\partial D_j}], \quad j=1,\dots ,N.
    \]

    The capacitance operator (matrix) $\mathcal{C}^{(N)}\in \R^{N\times N}$ associated with the $N$-resonator system is then given by
    \[
        \mathcal{C}^{(N)}_{ij} \coloneqq -\int_{\partial D_i}\phi_j\mathrm{d\sigma}
    \]
    and the \emph{material parameter matrix} $\mathcal{V}^{(N)}\in \R^{N\times N}$ by 
    \[
        \mathcal{V}^{(N)} \coloneqq \diag\left(\frac{v_1^2}{\abs{D_1}}, \dots, \frac{v_N^2}{\abs{D_N}}\right).
    \]
\end{definition}

We have the following equivalent characterisation of the capacitance matrix.
\begin{lemma}\label{lem:capcharact}
For $i = 1, \dots, N$, let $V_i^{(N)}$ be defined as the solution to the exterior boundary value problem
\[
\begin{cases}
\Delta V_i^{(N)} = 0 & \text{in } \mathbb{R}^3 \setminus \overline{\mathcal{D}^{(N)}}, \\
V_i^{(N)} = \mathbbm{1}_{\partial D_i} & \text{on } \partial \mathcal{D}^{(N)}, \\
V_i^{(N)}(\bm{x}) = \mc O(|\bm{x}|^{-1}) & \text{as } |\bm{x}| \to \infty.
\end{cases}
\]
Then, the capacitance coefficients, defined in \cref{def:capmat}, are given by
\[
\mathcal{C}^{(N)}_{ij} = \int_{\mathbb{R}^3 \setminus \overline{\mathcal{D}^{(N)}}} \nabla V_i^{(N)} \cdot \nabla V_j^{(N)} \, \mathrm{d}x,
\]
or equivalently
\begin{equation} \label{eq_finite_cap_pde_def}
\mathcal{C}^{(N)}_{ij}=-\int_{\partial D_j}\frac{\partial V_i^{(N)}}{\partial \nu}d\sigma
\end{equation}
for $1\leq i,j\leq N$. Here, $\nu$ denotes the outward unit normal on $\partial \mathcal{D}^{(N)}$.
\end{lemma}
The capacitance operator is valuable as their eigenvalues form the leading-order expansion of the subwavelength resonance frequencies; see \cite{cbms} for details.
\begin{theorem}[Capacitance operator approximation]\label{thm:capapprox}
For a system of $N$ resonators $D_1, \dots , D_N$ there exist exactly $N$ subwavelength resonant frequencies $\omega(\delta) > 0$, with the following asymptotics as $\delta\to 0$
\[
    \omega_n = \sqrt{\delta\lambda_n} + \mc O(\delta), \quad n=1, \dots, N,
\]
where $\lambda_1, \dots, \lambda_N$ are the eigenvalues of the generalised capacitance operator $\mathcal{V}^{(N)}\mathcal{C}^{(N)}$. Moreover, the normalised resonant solution $u_n$ of \cref{eq:hheq} associated with $\omega_n$ is given asymptotically by 
\[
    u_n(\bm{x}) = \sum_{i=1}^N \bseq{u_n}{i}V_i^{(N)}(\bm{x}) +\mathcal{O}(\delta), \quad x\in \R^3, n=1, \dots, N,
\]
where $\bm u_n$ is the normalised eigenvector of the generalised capacitance operator $\mathcal{V}^{(N)}\mathcal{C}^{(N)}$ associated with $\lambda_n$.
\end{theorem}

Here we summarize the key properties of the finite capacitance operators:
\begin{lemma}\label{lem:capprop}
    The capacitance operator $\mathcal{C}^{(N)} \in \R^{N\times N}$
    \begin{enumerate}
        \item is symmetric; 
        \item is positive definite;
        \item has negative off-diagonal entries, \emph{i.e.} $\mathcal{C}^{(N)}_{ij}<0$ for $i\neq j$ and
        \item is diagonally dominant, \emph{i.e.} $\mathcal{C}^{(N)}_{ii} > \sum_{j\neq i}\abs{\mathcal{C}^{(N)}_{ij}}$.
    \end{enumerate}
\end{lemma}
Notably, while the capacitance operator $\mathcal{C}^{(N)}$ is symmetric, the generalised capacitance operator $\mathcal{V}^{(N)}\mathcal{C}^{(N)}$ is \emph{not}. However, by calculating the square root $(\mathcal{V}^{(N)})^{\frac{1}{2}}$, we may instead study a similar spectral problem for 
\begin{equation}\label{eq:simprob}
        (\mathcal{V}^{(N)})^{\frac{1}{2}}\mathcal{C}^{(N)}(\mathcal{V}^{(N)})^{\frac{1}{2}} \sim \mathcal{V}^{(N)}\mathcal{C}^{(N)}.
\end{equation}

We will investigate the large array limit where the number of resonators $N\to \infty$ when the system is translation invariant in one direction, i.e., it constitutes a chain of resonators.  To understand this limit, for simplicity, instead of considering chains of $N$ resonators, we let the number of resonators be $2N+1$ and identify the capacitance matrix $\mathcal{C}^{(2N+1)}\in \mathbb{R}^{(2N+1)\times (2N+1)}$ with the lattice operator on $\mathbb{Z}$:
\begin{equation} \label{eq_finite_cap_def}
\begin{aligned}
\mathcal{C}^{(2N+1)}:\lz \to \lz,\quad \mathcal{C}^{(2N+1)}:=\iota_N \circ \mathcal{C}^{(2N+1)} \circ \iota^*_N .
\end{aligned} 
\end{equation}
Here, with a little abuse of notation, $\mathcal{C}^{(2N+1)}$ denotes both the finite matrix defined in Lemma \ref{lem:capcharact} and its identification as an operator acting on $\lz$. The operator $\iota_N:\mathbb{R}^{2N+1}\to \lz$ denotes the extension by zeros and $\iota^*_N$ is its adjoint.

We analogously identify the diagonal material parameter matrix $\mathcal{V}^{(2N+1)}\in \mathbb{R}^{(2N+1)\times (2N+1)} $ as a lattice diagonal operator on $\lz$.

\section{Capacitance operator for infinitely periodic systems} \label{sec:periodic}

As mentioned above, we consider the one-dimensional lattice $\Lambda:=\mathbb{Z}(a\bm{e}_1)$. Associated with $\Lambda$ is the infinite resonator chain (see Figure \ref{fig_infinite_resonator})
\begin{equation*}
\mathcal{D}=\bigcup_{i\in\mathbb{Z}}D_i .
\end{equation*}
As in the case of finite arrays discussed in the previous section, the subwavelength resonance problem associated with the infinitely periodic resonator structure $\mathcal{D}$ is also known to be governed by the capacitance operator to the leading order; see \cite{cbms}. However, unlike the capacitance operator for the finite arrays \cite{cbms,ammari.barandun.ea2024Exponentially,ammari.barandun.ea2025Universal} and full lattices (i.e. the crystal case where the resonators are periodically arranged with respect to a $n$-dimensional lattice in $\mathbb{R}^{n}$, which will be referred to as the $n$D-in-$n$D case; see \cite{cbms,ammari.barandun.ea2023Edge}), which are well-studied, the $1$D-in-$3$D capacitance operator for an infinite periodic chain is much less understood. As the main focus of this work, we aim to prove the fundamental properties of the $1$D-in-$3$D capacitance operator, including its well-definedness and its off-diagonal decay rate.

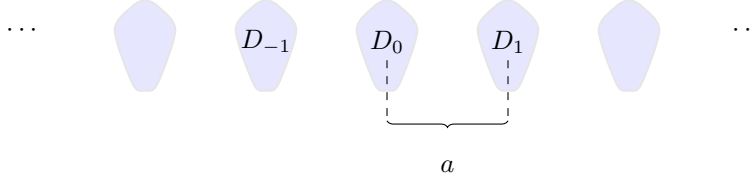
\begin{figure}
    \centering
    \begin{tikzpicture}[scale=0.8]
%resonators
\foreach \x in {-4,-2,0,2,4} { 
        \draw [thick,fill=blue,opacity=0.1] plot [smooth,samples=400, tension=1] coordinates {({0+\x},{0.5}) ({0.2+\x},{0.4}) ({0.5+\x},{0}) ({0.2+\x},{-0.9}) ({0+\x},{-1}) ({-0.2+\x},{-0.9}) ({-0.5+\x},{0}) ({-0.2+\x},{0.4}) ({0+\x},{0.5})};
    }
\node[scale=1,thick] at (6,0) {$\cdots$};
\node[scale=1,thick] at (-6,0) {$\cdots$};
%label of resonators
\node[scale=1,thick] at ({0},{-0.2}) {$D_0$};
\node[scale=1,thick] at ({-2},{-0.2}) {$D_{-1}$};
\node[scale=1,thick] at ({2},{-0.2}) {$D_{1}$};
%spacing
\draw[dashed] ({0},{-0.5})--({0},{-1.5});
\draw[dashed] ({2},{-0.5})--({2},{-1.5});
\draw[decorate,decoration={brace,mirror}] ({0},{-1.5}) -- ({2},{-1.5});
\node[below,scale=1] at (1,-2) {$a$};
\end{tikzpicture}
    \caption{An infinite periodic chain of resonators.}
    \label{fig_infinite_resonator}
\end{figure}

\subsection{Properties of the capacitance operator}

The kernel of the capacitance operator associated with the infinite periodic array is defined via the Floquet transform:
\begin{equation} \label{eq_cap_def}
\mathcal{C}(n,m):=\frac{a}{2\pi}\int_{-\pi/a}^{\pi /a}\widehat{\mathcal{C}}^{\alpha}e^{-i\alpha(n-m)a}d\alpha,
\end{equation}
where the Floquet component $\widehat{\mathcal{C}}^{\alpha}$ is defined as
\begin{equation} \label{eq_cap_momentum_def}
\widehat{\mathcal{C}}^{\alpha}:=-\int_{\partial D}(\mathcal{S}_{D}^{\alpha})^{-1}[\mathbbm{1}_{\partial D}]d\sigma
\end{equation}
with $\mathcal{S}_{D}^{\alpha}:L^2(\partial D)\to H^1(\partial D)$ being the quasi-periodic single-layer potential operator
\begin{equation} \label{eq_SL_momentum_def}
\mathcal{S}_{D}^{\alpha}[\varphi](\bm{x})=-\frac{1}{4\pi}\int_{\partial D}d\sigma_{\bm{y}}\varphi(\bm{y})\Big[\sum_{n\in\mathbb{Z}} \frac{e^{i\alpha na}}{|\bm{x}-\bm{y}-na\bm{e}_1|}\Big]. 
\end{equation}

To rigorously justify that \eqref{eq_cap_def} is well-defined, we establish two results on the regularity of the Floquet component $\widehat{\mathcal{C}}^{\alpha}$.
\begin{theorem} \label{thm_cap_def_momentum_existence}
$\widehat{\mathcal{C}}^{\alpha}$ is smooth for $\alpha\in [-\pi/a,0)\cup(0,\pi/a]$.
\end{theorem}

\begin{theorem} \label{thm_cap_def_momentum_estimate}
There exists $\alpha_0 \in (0,1/a)$ such that for $0<|\alpha|<\alpha_0$, $\widehat{\mathcal{C}}^{\alpha}$ admits the decomposition
\begin{equation} \label{eq_cap_momentum_decompose_1}
\widehat{\mathcal{C}}^{\alpha}=-|\partial D|\kappa_{*}^{-1}(\alpha)+r(\alpha)
\end{equation}
with
\begin{equation*}
\kappa_{*}(\alpha)=\frac{1}{2\pi a}\log\big(2\sin \frac{a|\alpha|}{2}\big),
\end{equation*}
and the remainder $r(\alpha)$ satisfying the following estimate: 
\begin{equation} \label{eq_cap_momentum_remainder_1}
r(\alpha)=\mathcal{O}(\frac{1}{\log^2 |\alpha|}) .
\end{equation}
Moreover, the derivatives of the remainder satisfy
\begin{equation} \label{eq_cap_momentum_remainder_2}
\partial_{\alpha}r(\alpha)=\mathcal{O}(\frac{1}{\alpha\log^3 |\alpha|}),\quad \partial_{\alpha}^2 r(\alpha)=\mathcal{O}(\frac{1}{\alpha^2\log^3 |\alpha|}).
\end{equation}
\end{theorem}

Throughout this section, we assume that $a=|\partial D|=1$ (i.e., unit lattice spacing and surface area of the resonator $D$) for simplicity.

\subsubsection{Proof of Theorem \ref{thm_cap_def_momentum_existence}}
To establish Theorem \ref{thm_cap_def_momentum_existence}, it suffices to prove that the single-layer potential operator $\mathcal{S}_{D}^{\alpha}$ is smooth and invertible for $\alpha\neq 0$. To this end, we follow the standard strategy by decomposing $\mathcal{S}_{D}^{\alpha}$ into an invertible part and a compact part, and then applying the Fredholm alternative. More precisely, we write
\begin{equation} \label{eq_cap_def_momentum_existence_proof_1}
\begin{aligned}
\mathcal{S}_{D}^{\alpha}[\varphi](\bm{x})
&=-\frac{1}{4\pi}\int_{\partial D}d\sigma_{\bm{y}}\varphi(\bm{y}) \frac{1}{|\bm{x}-\bm{y}|}  -\frac{1}{4\pi}\int_{\partial D}d\sigma_{\bm{y}}\varphi(\bm{y})\Big[\sum_{n\neq 0} \frac{e^{i\alpha n}}{|\bm{x}-\bm{y}-n\bm{e}_1|}\Big]  \\
&=:\mathcal{S}_{D}^{free}[\varphi](\bm{x})+\mathcal{K}_{D}^{\alpha}[\varphi](\bm{x}) .
\end{aligned}
\end{equation}
We will first prove that the operator $\mathcal{K}_{D}^{\alpha}$ is compact and smooth:
\begin{proposition} \label{prop_cap_def_momentum_existence_proof_1}
Let $K^{\alpha}(\bm{x},\bm{y})=\sum_{n\neq 0} \frac{e^{i\alpha n}}{|\bm{x}-\bm{y}-n\bm{e}_1|}$.  For any $\alpha\in [-\pi,0)\cup(0,\pi]$, $K^{\alpha}\in C^1(\partial D\times \partial D)$. Moreover, the map $\alpha\mapsto K^{\alpha}$ is smooth for $\alpha\in [-\pi,0)\cup(0,\pi]$.
\end{proposition}
By Proposition \ref{prop_cap_def_momentum_existence_proof_1} and the fact that $\mathcal{S}_{D}^{free}:L^2(\partial D)\to H^1(\partial D)$ is invertible, we know $\mathcal{S}_{D}^{\alpha}$ is Fredholm for $\alpha\neq 0$. Hence, now only it remains to prove that $\mathcal{S}_{D}^{\alpha}$ is injective to conclude the proof of Theorem \ref{thm_cap_def_momentum_existence}. This is deduced from the following proposition.
\begin{proposition} \label{prop_cap_def_momentum_existence_proof_2}
Let $\varphi\in L^2(\partial D)$. For $\bm{x}\in Y$, we define
\begin{equation*}
u(\bm{x}):=-\frac{1}{4\pi}\int_{\partial D}d\sigma_{\bm{y}}\varphi(\bm{y})\Big[\sum_{n\in \mathbb{Z}} \frac{e^{i\alpha n}}{|\bm{x}-\bm{y}-n\bm{e}_1|}\Big] .
\end{equation*}
Then, for $\bm{x}\in Y\backslash \partial D$, $u(\bm{x})$ is harmonic. The function $u$ is quasi-periodic in the sense that
\begin{equation} \label{eq_cap_def_momentum_existence_proof_4}
u(1,x_{\perp})=e^{i\alpha}u(0,x_{\perp}).
\end{equation}
Across the surface, $u$ is continuous, and the normal derivative $\partial_{\nu}u$ admits a nonzero jump:
\begin{equation} \label{eq_cap_def_momentum_existence_proof_2}
u\big|_{\partial D^{-}}=u\big|_{\partial D^{+}}=\mathcal{S}_{D}^{\alpha}[\varphi],\quad \partial_{\nu}u\big|_{\partial D^{-}}-\partial_{\nu}u\big|_{\partial D^{+}}=\varphi .
\end{equation}
Moreover, as $|x_{\perp}|\to \infty$, $u(x_1,x_{\perp})$ admits the following asymptotics:
\begin{equation} \label{eq_cap_def_momentum_existence_proof_3}
u(x_1,x_{\perp})=o\big(|x_{\perp}|^{-\frac{1}{2}}\big),\quad \nabla_{x_{\perp}}u(x_1,x_{\perp})=o\big(|x_{\perp}|^{-\frac{1}{2}}\big).
\end{equation}
\end{proposition}
We now apply Proposition \ref{prop_cap_def_momentum_existence_proof_2} to prove the injectivity of $\mathcal{S}_{D}^{\alpha}$. Suppose that $\mathcal{S}_{D}^{\alpha}[\varphi]=0$. Then, applying the Gauss-Green formula for $u(\bm{x})$ for $\bm{x}\in Y\backslash \overline{D}$ and $\bm{x}\in D$, respectively, yields
\begin{equation*}
\int_{Y\backslash \overline{D}}|\nabla u|^2=\int_{D}|\nabla u|^2=0.
\end{equation*}
This, with the continuity of Dirichlet data shown in \eqref{eq_cap_def_momentum_existence_proof_2}, implies $u(\bm{x})\equiv 0$ for $\bm{x}\in Y$. Then, using the jump formula \eqref{eq_cap_def_momentum_existence_proof_2}, we conclude that $\varphi=\partial_{\nu}u\big|_{\partial D^{-}}-\partial_{\nu}u\big|_{\partial D^{+}}=0$.

The proofs of Propositions \ref{prop_cap_def_momentum_existence_proof_1} and \ref{prop_cap_def_momentum_existence_proof_2} are given in \Cref{sec:app_A}.

\subsubsection{Proof of Theorem \ref{thm_cap_def_momentum_estimate}}

Before showing the proof, we first illustrate the irregularity of $\widehat{\mathcal{C}}^{\alpha}$ near $\alpha=0$. In fact, this follows from the definition of the single-layer potential \eqref{eq_SL_momentum_def}, which is clearly ill-defined for $\alpha=0$ due to the divergence of the harmonic series. Looking more closely at \eqref{eq_SL_momentum_def}, \textit{we observe that the divergence occurs only for constant functions in $L^2(\partial D)$}. In fact, for $\varphi \in L^2(\partial D)$ with $\int_{\partial D}\varphi =0$, we see that
\begin{equation} \label{eq_cap_def_momentum_estimate_proof_1}
\begin{aligned}
\mathcal{S}_{D}^{\alpha}[\varphi](\bm{x})&=-\frac{1}{4\pi}\int_{\partial D}d\sigma_{\bm{y}}\varphi(\bm{y})\Big[\sum_{n\in\mathbb{Z}} \frac{e^{i\alpha n}}{|\bm{x}-\bm{y}-n\bm{e}_1|}\Big]  \\
&=-\frac{1}{4\pi}\int_{\partial D}d\sigma_{\bm{y}}\varphi(\bm{y})\Big[\sum_{n\in\mathbb{Z}} \big(\frac{1}{|\bm{x}-\bm{y}-n\bm{e}_1|}-\frac{1}{|n\bm{e}_1|}\big)e^{i\alpha n}\Big],
\end{aligned}
\end{equation}
where no divergence appears as $\alpha\to 0$.

In light of this observation, we decompose $\mathcal{S}_{D}^{\alpha}$ into the two orthogonal components of $L^2(\partial D)$, i.e., the constant-valued and mean-zero functions, and apply the strategy based on the Schur complement to study the inverse $(\mathcal{S}_{D}^{\alpha})^{-1}$. To start, we introduce the projectors
\begin{equation*}
P_{\parallel }u:=(\int_{\partial D}u)\mathbbm{1}_{\partial D},\quad P_{\perp }=\mathbbm{1}-P_{\parallel }
\end{equation*}
and write
\begin{equation} \label{eq_schur_decomposition}
\mathcal{S}_{D}^{\alpha}
=\begin{pmatrix}
P_{\parallel }\mathcal{S}_{D}^{\alpha}P_{\parallel } & P_{\parallel }\mathcal{S}_{D}^{\alpha}P_{\perp } \\
P_{\perp }\mathcal{S}_{D}^{\alpha}P_{\parallel } & P_{\perp }\mathcal{S}_{D}^{\alpha}P_{\perp }
\end{pmatrix}
=:
\begin{pmatrix}
\mathcal{S}_{D}^{\alpha,11} & \mathcal{S}_{D}^{\alpha,12} \\
\mathcal{S}_{D}^{\alpha,21} & \mathcal{S}_{D}^{\alpha,22} 
\end{pmatrix}.
\end{equation}
The asymptotic behaviour of each block is illustrated as follows. In light of \eqref{eq_cap_def_momentum_estimate_proof_1}, we decompose $\mathcal{S}_{D}^{\alpha}$ as follows:
\begin{equation} \label{eq_cap_def_momentum_estimate_proof_2}
\begin{aligned}
\mathcal{S}_{D}^{\alpha}[\varphi](\bm{x})&=-\frac{1}{4\pi}\int_{\partial D}d\sigma_{\bm{y}}\varphi(\bm{y})\sum_{n\neq 0}\frac{e^{i\alpha n}}{|n\bm{e}_1|} \\
&\quad -\frac{1}{4\pi}\int_{\partial D}d\sigma_{\bm{y}}\varphi(\bm{y})\Big[\frac{1}{|\bm{x}-\bm{y}|}+\sum_{n\neq 0} \big(\frac{1}{|\bm{x}-\bm{y}-n\bm{e}_1|}-\frac{1}{|n\bm{e}_1|}\big)\Big]  \\
&\quad -\frac{1}{4\pi}\int_{\partial D}d\sigma_{\bm{y}}\varphi(\bm{y})\Big[\sum_{n\neq 0} \big(\frac{1}{|\bm{x}-\bm{y}-n\bm{e}_1|}-\frac{1}{|n\bm{e}_1|}\big)(e^{i\alpha n}-1)\Big]  \\
&=:\kappa_{*}(\alpha)P_{\parallel} +\mathcal{A}+\mathcal{R}^{\alpha}
\end{aligned}
\end{equation}
with
\begin{equation*}
\kappa_{*}(\alpha)=-\frac{1}{4\pi}\sum_{n\neq 0}\frac{e^{i\alpha n}}{|n\bm{e}_1|}= \frac{1}{2\pi}\log\big(2\sin \frac{|\alpha|}{2}\big) .
\end{equation*}
We will prove that \textit{the remainder $\mathcal{R}^{\alpha}$ is uniformly bounded as $\alpha\to 0$} (in contrast to the blowing-up term $\kappa_{*}(\alpha)$).
\begin{proposition} \label{prop_remainder_R_alpha}
The operator-valued map $\alpha\mapsto \mathcal{R}^{\alpha}\in \mathcal{B}(L^2(\partial D),H^1(\partial D))$ is smooth for $\alpha\in [-\pi,0)\cup(0,\pi]$. 
Here, $\mathcal{B}(X,Y)$, where $X$ and $Y$ are two Hilbert spaces, denotes the space of linear operators from $X$ into $Y$.
Moreover, as $\alpha\to 0$, it satisfies the following estimates.
\begin{equation} \label{eq_remainder_R_alpha_1}
\|\mathcal{R}^{\alpha}\|_{\mathcal{B}(L^2(\partial D),H^1(\partial D))}=\mathcal{O}(|\alpha| \log |\alpha|) ,
\end{equation}
\begin{equation} \label{eq_remainder_R_alpha_2}
\|\partial_{\alpha}\mathcal{R}^{\alpha}\|_{\mathcal{B}(L^2(\partial D),H^1(\partial D))}=\mathcal{O}(\log |\alpha|) ,
\end{equation}
\begin{equation} \label{eq_remainder_R_alpha_3}
\|\partial_{\alpha}^2 \mathcal{R}^{\alpha}\|_{\mathcal{B}(L^2(\partial D),H^1(\partial D))}=\mathcal{O}(|\alpha|^{-1}) .
\end{equation}
\end{proposition}
The proof of Proposition \ref{prop_remainder_R_alpha} is given in \Cref{sec:app_B}.

Let us first look at the second diagonal block $\mathcal{S}_{D}^{\alpha,22}$. By the decomposition \eqref{eq_cap_def_momentum_estimate_proof_2}, it is clear that
\begin{equation} \label{eq_cap_def_momentum_estimate_proof_3}
\mathcal{S}_{D}^{\alpha,22}=P_{\perp}(\kappa_{*}(\alpha)P_{\parallel}+\mathcal{A}+\mathcal{R}^{\alpha})P_{\perp}=P_{\perp}(\mathcal{A}+\mathcal{R}^{\alpha})P_{\perp},
\end{equation}
which indicates that the blowing-up term $\kappa_{*}(\alpha)$ is subtracted from the block $\mathcal{S}_{D}^{\alpha,22}$. Hence, Proposition \ref{prop_remainder_R_alpha} ensures that $\mathcal{S}_{D}^{\alpha,22}$ is continuous for $\alpha\in [-\pi,\pi]$. Moreover, it is invertible, as guaranteed by the following claim (proved in \Cref{sec:app_B}) and the estimate in Proposition \ref{prop_remainder_R_alpha}.
\begin{proposition} \label{prop_A_alpha_invertibility}
The operator $P_{\perp }\mathcal{A}P_{\perp }: P_{\perp }L^2(\partial D)\to P_{\perp }H^1(\partial D)$ is invertible.
\end{proposition}
Now, we consider the first diagonal block $\mathcal{S}_{D}^{\alpha,11}$. Again, by the decomposition \eqref{eq_cap_def_momentum_estimate_proof_2} and the estimate in Proposition \ref{prop_remainder_R_alpha}, one sees
\begin{equation} \label{eq_cap_def_momentum_estimate_proof_4}
\mathcal{S}_{D}^{\alpha,11}=P_{\parallel}(\kappa_{*}(\alpha)P_{\parallel}+\mathcal{A}+\mathcal{R}^{\alpha})P_{\parallel}=\kappa_{*}(\alpha)P_{\parallel}\big(\mathbbm{1}+\mathcal{O}(\frac{1}{\log |\alpha|})\big)P_{\parallel},
\end{equation}
which implies that $\mathcal{S}_{D}^{\alpha,11}$ is also invertible. Thus, by the method of the Schur complement, we know that $\mathcal{S}_{D}^{\alpha}$ is invertible with the following expression if $\mathcal{V}^{\alpha}:=\mathcal{S}_{D}^{\alpha,11}-\mathcal{S}_{D}^{\alpha,12}(\mathcal{S}_{D}^{\alpha,22})^{-1}\mathcal{S}_{D}^{\alpha,21}$ is invertible:
\begin{equation*}
(\mathcal{S}_{D}^{\alpha})^{-1}=
\begin{pmatrix}
(\mathcal{V}^{\alpha})^{-1} & -(\mathcal{V}^{\alpha})^{-1}\mathcal{S}_{D}^{\alpha,12}(\mathcal{S}_{D}^{\alpha,22})^{-1} \\
-(\mathcal{S}_{D}^{\alpha,22})^{-1}\mathcal{S}_{D}^{\alpha,21}(\mathcal{V}^{\alpha})^{-1} & (\mathcal{S}_{D}^{\alpha,22})^{-1}+(\mathcal{S}_{D}^{\alpha,22})^{-1} \mathcal{S}_{D}^{\alpha,21} (\mathcal{V}^{\alpha})^{-1} \mathcal{S}_{D}^{\alpha,12} (\mathcal{S}_{D}^{\alpha,22})^{-1}
\end{pmatrix}.
\end{equation*}
Indeed, $\mathcal{V}^{\alpha}$ is invertible for small $\alpha$. In fact, by \eqref{eq_cap_def_momentum_estimate_proof_2} and Proposition \ref{prop_remainder_R_alpha}, we know that
\begin{equation*}
\mathcal{S}_{D}^{\alpha,12/21}=P_{\parallel/\perp}(\mathcal{A}+\mathcal{R}^{\alpha})P_{\perp/\parallel}=\mathcal{O}(1).
\end{equation*}
Hence, for small $\alpha$, it follows that
\begin{equation} \label{eq_cap_def_momentum_estimate_proof_5}
\mathcal{S}_{D}^{\alpha,12}(\mathcal{S}_{D}^{\alpha,22})^{-1}\mathcal{S}_{D}^{\alpha,21}  (\mathcal{S}_{D}^{\alpha,11})^{-1}=\mathcal{O}(\|(\mathcal{S}_{D}^{\alpha,11})^{-1}\|)=\mathcal{O}(\frac{1}{\log |\alpha|}),
\end{equation}
and the inverse $(\mathcal{V}^{\alpha})^{-1}$ is given by
\begin{equation} \label{eq_cap_def_momentum_estimate_proof_6}
(\mathcal{V}^{\alpha})^{-1}=(\mathcal{S}_{D}^{\alpha,11})^{-1}\big(\mathbbm{1}-\mathcal{S}_{D}^{\alpha,12}(\mathcal{S}_{D}^{\alpha,22})^{-1}\mathcal{S}_{D}^{\alpha,21}  (\mathcal{S}_{D}^{\alpha,11})^{-1} \big)^{-1}.
\end{equation}
Furthermore, as shown in the next proposition, using \eqref{eq_cap_def_momentum_estimate_proof_3}-\eqref{eq_cap_def_momentum_estimate_proof_6}, Proposition \ref{prop_remainder_R_alpha}, and the resolvent identity, a detailed estimate of $(\mathcal{V}^{\alpha})^{-1}$ can be obtained. The proof is left to the interested reader.
\begin{proposition} \label{prop_inverse_constant_component}
There exists $\alpha_0 \in (0,1)$ such that for $0<|\alpha|<\alpha_0$, it holds that
\begin{equation}
(\mathcal{V}^{\alpha})^{-1}=\kappa_{*}^{-1}(\alpha)P_{\parallel}+h^{\alpha},
\end{equation}
where the remainder $h^{\alpha}$ admits the following estimates as $\alpha\to 0$:
\begin{equation} \label{eq_inverse_constant_component_1}
\|h^{\alpha}\|_{\mathcal{B}(H^1(\partial D),L^2(\partial D))}=\mathcal{O}(\frac{1}{\log^2 |\alpha|}),
\end{equation}
\begin{equation} \label{eq_inverse_constant_component_2}
\|\partial_{\alpha}h^{\alpha}\|_{\mathcal{B}(H^1(\partial D),L^2(\partial D))}=\mathcal{O}(\frac{1}{|\alpha|\log^3 |\alpha|}),
\end{equation}
\begin{equation} \label{eq_inverse_constant_component_3}
\|\partial_{\alpha}^2 h^{\alpha}\|_{\mathcal{B}(H^1(\partial D),L^2(\partial D))}=\mathcal{O}(\frac{1}{|\alpha|^2\log^3 |\alpha|}).
\end{equation}
\end{proposition}
With Proposition \ref{prop_inverse_constant_component}, Theorem \ref{thm_cap_def_momentum_estimate} follows directly by noting that
\begin{equation*}
\widehat{\mathcal{C}}^{\alpha}=-\int_{\partial D}(\mathcal{S}_{D}^{\alpha})^{-1}[\mathbbm{1}_{\partial D}]
=-\int_{\partial D}(\mathcal{V}^{\alpha})^{-1}[\mathbbm{1}_{\partial D}]
=-\kappa_{*}^{-1}(\alpha)+r(\alpha)
\end{equation*}
with $$
r(\alpha):=-\int_{\partial D}h^{\alpha}[\mathbbm{1}_{\partial D}] .
$$

\subsection{Decay estimate of the long-range interactions}\label{sec:LRI}

\subsubsection{Statement of the result}
Unlike the full-lattice capacitance operator, which exhibits exponential off-diagonal decay \cite{ammari2025nonlinear_soliton,ammari2026resolvent_convergence}, the $1$D-in-$3$D capacitance operator is expected to host only a polynomial off-diagonal decay rate due to the additional transverse degrees of freedom beyond the lattice direction, as will be shown in the next theorem. This is manifested by the blowing up of $\partial_{\alpha}\widehat{\mathcal{C}}^{\alpha}$, as shown in Theorem \ref{thm_cap_def_momentum_estimate}. Using the estimate obtained in Theorem \ref{thm_cap_def_momentum_estimate}, we confirm that the $1$D-in-$3$D capacitance operator indeed possesses a polynomial off-diagonal decay rate.
\begin{theorem} \label{thm_off_diag_decay}
There exists $C\in\mathbb{R},N>0$ such that for $n,m\in\mathbb{Z}$ with $|n-m|>N$, it holds that
\begin{equation*}
\mathcal{C}(n,m)=\frac{C}{|n-m|\log^2|n-m|}+o\left(\frac{1}{|n-m|\log^2|n-m|}\right).
\end{equation*}
\end{theorem}
We briefly comment on this result. The most interesting feature of the decay rate obtained in Theorem \ref{thm_off_diag_decay} is that, although $\mathcal{C}(n,m)$ decays slower than $1/|n-m|^p$ for any $p>1$, \textit{it is summable along the off-diagonals}, i.e.,
\begin{equation} \label{eq_column_sum}
\sum_{m\in\mathbb{Z}}|\mathcal{C}(n,m)|<\infty .
\end{equation}
The property \eqref{eq_column_sum} has implications, for example, in Anderson localisation. As proved in \cite{aizenman.molchanov1993Localization,shi2021multiscale}, a discrete operator with finite off-diagonal sum exhibits only a pure point spectrum when a large enough on-site i.i.d. random perturbation is applied. In particular, we expect more for our 1D-like operator, since physically, for one-dimensional systems, a random disorder with arbitrarily small amplitude should result in a purely point spectrum. Nevertheless, to our knowledge, this is only proved for discrete operators with an off-diagonal decay rate faster than $1/|n-m|^8$ \cite{Molchanov1999localization_1d_long_range} (see also \cite{shi2023almost_periodic_long_range,jian2025localization_interaction_long_range}), which excludes our borderline decay rate. We hope that our concrete physical example can provide an additional motivation for experts in this field.

We note that, in the statement of Theorem \ref{thm_off_diag_decay}, the capacitance operator itself does not encode any information about the material parameters inside the resonators, which is important for manipulating waves in the subwavelength regime. (For example, tuning the material parameter locally may confine the wave in the defected resonator, leading to the appearance of defect-localised modes.) In fact, similar to Section 2, if one takes the material parameter into account, the governing operator of the discrete resonance problem will become (see \eqref{eq:simprob})
\begin{equation*}
\mathcal{V} \mathcal{C} \quad \text{or}\quad \mathcal{V}^{\frac{1}{2}}\mathcal{C} \mathcal{V}^{\frac{1}{2}}
\end{equation*}
with $\mathcal{V}$ being a positive multiplication operator on $\ell^2(\mathbb{Z})$. As $\mathcal{V}$ is diagonal, the off-diagonal decay rate of $V\mathcal{C}$ or $V^{\frac{1}{2}}\mathcal{C} \mathcal{V}^{\frac{1}{2}}$ is dominated by the polynomial decay rate shown in Theorem \ref{thm_off_diag_decay}. This means that the off-diagonal estimate in Theorem \ref{thm_off_diag_decay} holds also for the generalised capacitance operator, including non-uniform distribution of material parameters.

\subsubsection{Proof of Theorem \ref{thm_off_diag_decay}}

Based on Theorem \ref{thm_cap_def_momentum_estimate}, we prove Theorem \ref{thm_off_diag_decay} by a standard asymptotic analysis of oscillating integrals. By periodicity, it suffices to prove Theorem \ref{thm_off_diag_decay} for $m=0$ and $n>0$.

To start, we note that the Hermiticity of the single-layer potential $\mathcal{S}_{D}^{\alpha}$ implies that the Floquet component $\widehat{\mathcal{C}}^{\alpha}$ is real and symmetric.
\begin{lemma} \label{lem_real_even_cap_momentum}
For $\alpha\in [-\pi,0)\cup (0,\pi]$, we have $\widehat{\mathcal{C}}^{\alpha}= \widehat{\mathcal{C}}^{-\alpha}\in\mathbb{R}$.
\end{lemma}
The detailed proof is given in \Cref{sec:app_C}. With this property, we see that
\begin{equation*}
\mathcal{C}(n,0)=\frac{1}{2\pi}\int_{-\pi}^{\pi}\widehat{\mathcal{C}}^{\alpha}e^{-i\alpha n}d\alpha
=\frac{1}{\pi}\int_{0}^{\pi}\widehat{\mathcal{C}}^{\alpha}\cos (n\alpha)d\alpha.
\end{equation*}
Select a smooth function $\eta\in C^{\infty}([0,\pi])$ such that
\begin{equation*}
\eta(\alpha)=1\quad \text{ for $\alpha\in [0,\alpha_{0}/2]$},\quad \eta(\alpha)=0\quad \text{ for $\alpha\in [\alpha_{0},\pi]$}
\end{equation*}
with $\alpha_0>0$ introduced in Theorem \ref{thm_cap_def_momentum_estimate}. Then, we further decompose the integral and integrate by parts
\begin{equation} \label{eq_off_diag_decay_proof_1}
\begin{aligned}
\mathcal{C}(n,0)&=\frac{1}{\pi}\int_{0}^{\pi}\eta(\alpha)\widehat{\mathcal{C}}^{\alpha}\cos (n\alpha)d\alpha + \frac{1}{\pi}\int_{0}^{\pi}(1-\eta(\alpha))\widehat{\mathcal{C}}^{\alpha}\cos (n\alpha)d\alpha \\
&=-\frac{1}{n\pi}\int_{0}^{\alpha_0/2}(\partial_{\alpha}\widehat{\mathcal{C}}^{\alpha})\sin (n\alpha)d\alpha \\
&\quad -\frac{1}{n\pi}\int_{\alpha_0/2}^{\alpha_0} \big[\eta(\alpha)(\partial_{\alpha}\widehat{\mathcal{C}}^{\alpha})+\partial_{\alpha}\eta(\alpha)\cdot 
\widehat{\mathcal{C}}^{\alpha}\big]\sin (n\alpha)d\alpha \\
&\quad -\frac{1}{n\pi}\int_{\alpha_0/2}^{\pi} \partial_{\alpha}\big[(1-\eta(\alpha)) \widehat{\mathcal{C}}^{\alpha} \big] \sin (n\alpha)d\alpha \\
&=:-\frac{1}{n\pi} \big[I_1(n)+I_2(n)+I_3(n)\big].
\end{aligned}
\end{equation}
Note that, by Theorem \ref{thm_cap_def_momentum_existence}, $(1-\eta(\alpha)) \widehat{\mathcal{C}}^{\alpha}$ is smooth on $[\alpha_0/2,\pi]$. Hence, a further integration by parts shows that
\begin{equation} \label{eq_off_diag_decay_proof_2}
I_3(n)=\mathcal{O}(n^{-1})
\end{equation}
as $n\to \infty$. Similarly, we have
\begin{equation} \label{eq_off_diag_decay_proof_3}
I_2(n)=\mathcal{O}(n^{-1}).
\end{equation}
Our main focus will be on $I_1(n)$, i.e., the integral near the origin. Let $n$ be large enough so that $n^{-1}<n^{-2/3}<\alpha_0/2$. We further decompose $I_1(n)$ as follows:
\begin{equation} \label{eq_off_diag_decay_proof_4}
\begin{aligned}
I_1(n)
&=-\int_{0}^{n^{-2/3}}(\partial_{\alpha}\kappa_*^{-1}(\alpha))\sin (n\alpha)d\alpha + \int_{0}^{n^{-1}}(\partial_{\alpha}r(\alpha))\sin (n\alpha)d\alpha \\
&\quad -\int_{n^{-2/3}}^{\alpha_0/2}(\partial_{\alpha}\kappa_*^{-1}(\alpha))\sin (n\alpha)d\alpha + \int_{n^{-1}}^{\alpha_0/2}(\partial_{\alpha}r(\alpha))\sin (n\alpha)d\alpha \\
&=:I_{1,near}^{(0)}(n) + I_{1,near}^{(1)}(n) + I_{1,far}^{(0)}(n) + I_{1,far}^{(1)}(n),
\end{aligned}
\end{equation}
where $\kappa_{*},r$ are introduced in Theorem \ref{thm_cap_def_momentum_estimate}. Recalling that $\kappa_{*}(\alpha)$ admits an explicit expression, one can directly check that the following result holds.
\begin{proposition} \label{prop_off_diag_decay_proof_1}
There exists $C_1\in\mathbb{R},N_1>0$ such that for $n> N_1$, it holds
\begin{equation*}
I_{1,near}^{(0)}(n)=\frac{C_1}{\log^2|n|} + \mathcal{O}\big(\frac{1}{\log^3|n|}\big),\quad
I_{1,far}^{(0)}(n)=\mathcal{O}\big(\frac{1}{\log^3|n|}\big).
\end{equation*}
\end{proposition}
The remaining integrals involving $r(\alpha)$ are estimated using the asymptotics obtained in Theorem \ref{thm_cap_def_momentum_estimate}. The following result holds. 
\begin{proposition} \label{prop_off_diag_decay_proof_2}
For any $p\in (0,1)$, there exists $N_2>0$ such that for $n> N_2$, it holds
\begin{equation*}
I_{1,near}^{(1)}(n)+I_{1,far}^{(1)}(n)\leq C_2\big(\frac{1-p^2}{p^2}\frac{1}{\log^2|n|}+\frac{1}{\log^3|n|}+ \frac{1}{|n|^{1-p}}\big),
\end{equation*}
where $C_2>0$ is independent of $p$.
\end{proposition}
The detailed proof of this proposition is presented in \Cref{sec:app_C}. By \eqref{eq_off_diag_decay_proof_4}, Propositions \ref{prop_off_diag_decay_proof_1} and \ref{prop_off_diag_decay_proof_2}, we select $p\in (0,1)$ so that
\begin{equation*}
C_2\frac{1-p^2}{p^2}<C_1.
\end{equation*}
The proof is then complete.

\begin{remark}
The approach of this section can be readily generalised to study the off-diagonal decay rate of the $1$D-in-$2$D (or $2$D-in-$3$D) capacitance operators, which requires more careful asymptotic analysis because the behaviour of the free-space Green function for the $1$D-in-$2$D system (Hankel function) is more delicate than $1/|\bm{x}-\bm{y}|$. We leave this extension to the interested reader.

One may also expect to include the case where multiple resonators are contained in the unit strip $Y$, namely $\cup_{1\leq k\leq K}D_k\Subset Y$, in which case the Floquet component of the capacitance operator is a $K\times K$ matrix defined as
\begin{equation*}
\widehat{\mathcal{C}}^{\alpha}_{i,j}:=\int_{\partial D_j}(\mathcal{S}_{D}^{\alpha})^{-1}[\mathbbm{1}_{\partial D_j}]d\sigma ,
\end{equation*}
and $\mathcal{C}_{i,j}$ is defined via the inverse Floquet transform. However, we should point out that the proof in this paper cannot directly apply to that case. The difficulty is due to the fact that the constant part of the decomposition \eqref{eq_cap_def_momentum_estimate_proof_2} (i.e., the part that is proportional to $\kappa_{*}(\alpha)$) is of rank $1$ instead of rank $K$. This leads to the problem that the first diagonal in \eqref{eq_schur_decomposition} is not invertible, and therefore prevents the application of the Schur complement to study the asymptotic behaviour of $\widehat{\mathcal{C}}^{\alpha}$. We leave this problem for future study.
\end{remark}

\section{Strong convergence of the finite capacitance operator and the implication for the spectrum}\label{sec:SC}
In this section, we first establish the strong convergence of the finite capacitance operator to its infinite counterpart as the chain size goes to infinity. Then, with this operator convergence, we discuss the convergence of the spectrum associated with the truncated system to that of the infinite structure, which shall recover the convergence results presented in \cite{ammari.davies.ea2023Convergence, ammari.davies.ea2025Spectral}. We also make use of  the decay rate of the off-diagonal entries of the capacitance operator to estimate the  convergence rate 
for the defect modes.  Our estimate is the first rigorous estimate of the rate of convergence. Nevertheless, it is weaker than the effective rate observed numerically in \cite{ammari.davies.ea2025Spectral, ammari.davies.ea2023Convergence}. Note that estimating the rate of convergence of the resonant frequencies of a system of coupled resonators in a truncated periodic lattice to the essential spectrum of the corresponding infinite lattice is more challenging due to delocalisation.

\subsection{Strong convergence of operator}
The main result of this section is as follows.
\begin{proposition} \label{prop_str_con}
    Let $\mathcal{C}^{(2N+1)}:\ell^2(\Z)\to \ell^2(\Z)$ be the finite capacitance operator introduced in \eqref{eq_finite_cap_def}, and let $\mathcal{C}$ be the capacitance operator associated with the infinite chain. Then 
    \[
        \mathcal{C}^{(2N+1)} \to \mathcal{C} \text{ strongly,}\quad \text{as }N\to \infty.
    \]
\end{proposition}

\begin{proof}
We first outline the main idea of the proof with the details presented in several steps. First, we claim that the diagonal elements of $\mathcal{C}^{(2N+1)}$ are uniformly bounded in the following sense:
\begin{equation} \label{eq_str_con_proof_1}
\mathcal{C}^{(2N+1)}(i,i)\leq \mathcal{C}(i,i),\quad \text{for all $N\geq 1$ and $-N\leq i\leq N$} .
\end{equation}
Hence, recalling that $\mathcal{C}^{(2N+1)}$ is diagonally dominant (see Lemma \ref{lem:capprop}), we conclude by the Schur test that the operator norm $\|\mathcal{C}^{(2N+1)}\|_{\mathcal{B}(\lz)}$ is uniformly bounded in $N$. This implies that, using the Banach-Steinhaus theorem, it is sufficient to prove the following to conclude the proof of Proposition \ref{prop_str_con}:
\begin{equation}\label{eq_str_con_proof_2}
(\mathcal{C}^{(2N+1)}-\mathcal{C})\bm{v}_{i}\to 0,\quad \text{for all $i\in\mathbb{Z}$} ,
\end{equation}
with $\bm{v}_i$ being the $i$\textsuperscript{th} unit basis. Without loss of generality, we show only the idea for the case $i=0$. Indeed, by direct calculation,
we have
\begin{equation*}
\begin{aligned}
\|(\mathcal{C}^{(2N+1)}-\mathcal{C})\bm{v}_{0}\|_{\lz}^2
&=\sum_{|i|>N}|\mathcal{C}(i,0)|^2 + \sum_{|i|<N}|\mathcal{C}^{(2N+1)}(i,0)-\mathcal{C}(i,0)|^2 \\
&\leq \sum_{|i|>N}|\mathcal{C}(i,0)|^2 + \big(\sum_{|i|<N}|\mathcal{C}^{(2N+1)}(i,0)-\mathcal{C}(i,0)|\big)^2 .
\end{aligned}
\end{equation*}
By the absolute convergence of infinite capacitance elements, as established in Theorem \ref{thm_off_diag_decay}, the first sum converges to zero as $N\to\infty$. The convergence of the second sum is proved using the dominated convergence theorem. In fact, note that
\begin{equation*}
\sum_{|i|<N}|\mathcal{C}^{(2N+1)}(i,0)-\mathcal{C}(i,0)|\leq 
\sum_{|i|<N}|\mathcal{C}^{(2N+1)}(i,0)|+\sum_{i\in\mathbb{Z}}|\mathcal{C}(i,0)|
\overset{(i)}{\leq }\mathcal{C}(i,i)+\sum_{i\in\mathbb{Z}}|\mathcal{C}(i,0)| <\infty ,
\end{equation*}
where the diagonal dominance of $\mathcal{C}^{(N)}$ and \eqref{eq_str_con_proof_1} is applied in (i). Hence, by the dominated convergence theorem, it suffices to show the following pointwise convergence to conclude the proof: 
\begin{equation} \label{eq_str_con_proof_3}
\mathcal{C}^{(2N+1)}(i,j)-\mathcal{C}(i,j) \to 0 ,\quad \text{for all $-N\leq i,j\leq N$} .
\end{equation}
In the sequel, we prove \eqref{eq_str_con_proof_1} and \eqref{eq_str_con_proof_3},  separately in two steps.

{\color{blue}Step 1.} The proof of \eqref{eq_str_con_proof_1} is based on the maximum principle of harmonic functions, as in \cite[Lemma 1.13]{cbms}. In fact, similar to \eqref{eq_finite_cap_pde_def}, we note that $\mathcal{C}(i,j)$ admits the following expression:
\begin{equation*} 
\mathcal{C}(i,j)=-\int_{\partial D_j}\frac{\partial V_i}{\partial \nu}d\sigma,\quad -N\leq i,j\leq N,
\end{equation*}
where $V_i$ solves the following equation:
\begin{equation*}
\left\{
\begin{aligned}
&\Delta V_i = 0\quad  \text{in } \mathbb{R}^3 \setminus \overline{\mathcal{D}}, \\
&V_i = \mathbbm{1}_{\partial D_i}\quad  \text{on } \partial \mathcal{D}, \\
&V_i(x) \to 0\quad  \text{as } |x| \to \infty .
\end{aligned}
\right.
\end{equation*}
This implies, for any $-N\leq i,j\leq N$,
\begin{equation} \label{eq_str_con_proof_4}
\mathcal{C}^{(2N+1)}(i,j)-\mathcal{C}(i,j)=-\int_{\partial D_j}\frac{\partial w_i^{(N,N)}}{\partial \nu}d\sigma,
\end{equation}
where $w_i^{(N,N+k)}$ ($k\geq 0$) solves
\begin{equation} \label{eq_str_con_proof_5}
\left\{
\begin{aligned}
&\Delta w_i^{(N,N+k)} = 0\quad  \text{in } \mathbb{R}^3 \setminus \overline{\mathcal{D}}, \\
&w_i^{(N,N+k)}|_{\partial D_j} = 0, \quad\text{for $|j|\leq N+k$, and } w_i^{(N,N+k)}|_{\partial D_j} = V_i^{(N)}|_{\partial D_j}, \quad\text{for $|j|>N+k$}, \\
&w_i^{(N,N+k)}(x) \to 0\quad  \text{as } |x| \to \infty .
\end{aligned}
\right.
\end{equation}

\begin{figure}
    \centering
    \begin{tikzpicture}[scale=0.6]
%resonators
\foreach \x in {-8,-6,-2,0,2,6,8} { 
        \draw [thick,fill=blue,opacity=0.1] plot [smooth,samples=400, tension=1] coordinates {({0+\x},{0.5}) ({0.2+\x},{0.4}) ({0.5+\x},{0}) ({0.2+\x},{-0.9}) ({0+\x},{-1}) ({-0.2+\x},{-0.9}) ({-0.5+\x},{0}) ({-0.2+\x},{0.4}) ({0+\x},{0.5})};
    }
\node[scale=1,thick] at (4,0) {$\cdots$};
\node[scale=1,thick] at (-4,0) {$\cdots$};
\node[scale=1,thick] at (10,0) {$\cdots$};
\node[scale=1,thick] at (-10,0) {$\cdots$};
%label of resonators
\node[scale=0.5,thick] at ({-8},{-0.2}) {$D_{-N-k-1}$};
\node[scale=0.5,thick] at ({-6},{-0.2}) {$D_{-N-k}$};
\node[scale=0.5,thick] at ({-2},{-0.2}) {$D_{-1}$};
\node[scale=0.5,thick] at ({0},{-0.2}) {$D_{0}$};
\node[scale=0.5,thick] at ({2},{-0.2}) {$D_{1}$};
\node[scale=0.5,thick] at ({6},{-0.2}) {$D_{N+k}$};
\node[scale=0.5,thick] at ({8},{-0.2}) {$D_{N+k+1}$};
%boundary trace
\node[scale=0.7,thick,above] at ({-8},{0.8}) {$V_i|_{\partial D_{-N-k-1}}$};
\node[scale=0.7,thick,above] at ({-6},{0.8}) {$0$};
\node[scale=0.7,thick,above] at ({-2},{0.8}) {$0$};
\node[scale=0.7,thick,above] at ({0},{0.8}) {$0$};
\node[scale=0.7,thick,above] at ({2},{0.8}) {$0$};
\node[scale=0.7,thick,above] at ({6},{0.8}) {$0$};
\node[scale=0.7,thick,above] at ({8},{0.8}) {$V_i|_{\partial D_{N+k+1}}$};
\end{tikzpicture}
    \caption{The value above the chain refers to the trace of $w_i^{(N,N+k)}$ on the respective resonator. By construction, $w_i^{(N,N+k)}$ vanishes on the surface of the first $2(N+k)+1$ resonators.}
    \label{fig_trace_aux_function}
\end{figure}
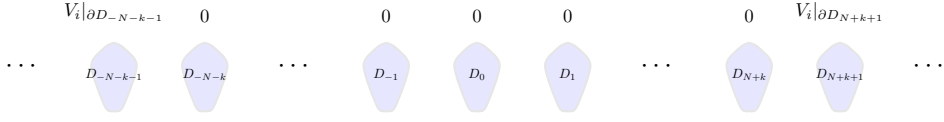

The trace of $w_i^{(N,N+k)}$ on $\partial \mathcal{D}$ is sketched in Figure \ref{fig_trace_aux_function} for illustration. In particular, we note that
\begin{equation*}
w_i^{(N,N+k)}\big|_{\partial D_j}-w_i^{(N,N+k+1)}\big|_{\partial D_j}=\left\{
\begin{aligned}
&V_i^{(N)}|_{\partial D_j}\in \big[0,\min\big\{1,\frac{C}{|j|}\big\}\big],\quad |j|=N+k, \\
&0,\quad \text{otherwise},
\end{aligned}
\right.
\end{equation*}
for some constant $C>0$ depending only on $N$. By the maximum principle, this implies that
\begin{equation*}
\frac{\partial w_i^{(N,N+k)}}{\partial \nu}\big|_{\partial D_j}
\geq \frac{\partial w_i^{(N,N+k+1)}}{\partial \nu}\big|_{\partial D_j}.
\end{equation*}
Hence, by \eqref{eq_str_con_proof_4}, it follows that
\begin{equation} \label{eq_str_con_proof_6}
\mathcal{C}^{(2N+1)}(i,j)-\mathcal{C}(i,j)\leq -\lim_{k\to\infty} \int_{\partial D_j}\frac{\partial w_i^{(N,N+k)}}{\partial \nu}d\sigma .
\end{equation}
As $k\to\infty$, it is clear that the normal derivative $\frac{\partial w_i^{(N,N+k)}}{\partial \nu}$ on the fixed boundary $\partial D_i$ converges to zero since $\text{dist}(D_i,\cup_{|j|> N+k} D_j)\to\infty$. Therefore, the right side of \eqref{eq_str_con_proof_6} is equal to zero, which leads to \eqref{eq_str_con_proof_1} by taking $j=i$.

{\color{blue}Step 2.} The proof of \eqref{eq_str_con_proof_3} is based on the same strategy. From \eqref{eq_str_con_proof_5} and the fact that $V_i^{(N)}|_{\partial D_j}<1$ for $|j|\geq N+k$, we know that
\begin{equation*}
w_i^{(N,N+k)}\Big|_{\partial \mathcal{D}} \leq \Big(w_i^{(N,N+k+1)}+V_{N+k}+V_{-N-k}\Big)\Big|_{\partial \mathcal{D}}.
\end{equation*}
Hence, applying the maximum principle, we obtain
\begin{equation*}
\frac{\partial w_i^{(N,N+k)}}{\partial \nu}\big|_{\partial D_j}
\leq \frac{\partial w_i^{(N,N+k+1)}}{\partial \nu}\big|_{\partial D_j}+\frac{\partial V_{N+k}}{\partial \nu}\big|_{\partial D_j}+\frac{\partial V_{-N-k}}{\partial \nu}\big|_{\partial D_j}.
\end{equation*}
Compared with \eqref{eq_str_con_proof_6}, this provides a lower bound of $\mathcal{C}^{(2N+1)}(i,j)-\mathcal{C}(i,j)$:
\begin{equation} \label{eq_str_con_proof_7}
\mathcal{C}^{(2N+1)}(i,j)-\mathcal{C}(i,j)\geq \sum_{|p|\geq N} \mathcal{C}(p,j) .
\end{equation}
Note that the right side of \eqref{eq_str_con_proof_7} converges to zero as $N\to\infty$, thanks to Theorem \ref{thm_off_diag_decay}. This estimate, together with \eqref{eq_str_con_proof_6}, completes the proof of \eqref{eq_str_con_proof_3}.
\end{proof}

\begin{remark}
We note that the strong convergence, as presented in Proposition \ref{prop_str_con}, is sharp. In fact, choosing test functions $u_N$ supported on $\mathbb{Z}\backslash \{-N,\cdots,N\}$ such that $\|\mathcal{C}u_N\|_{\lz}\equiv 1$, which exists since $\mathcal{C}$ is periodic, we observe
\begin{equation*}
\|(\mathcal{C}-\mathcal{C}^{(2N+1)})u_N\|_{\lz}=\|\mathcal{C}u_N\|_{\lz}\equiv 1 .
\end{equation*}
This implies that the operator $\mathcal{C}^{(2N+1)}$ does not converge to $\mathcal{C}$ in the operator norm; in other words, the convergence in the strong operator topology is the best that can be expected.
\end{remark}

\subsection{Convergence of the spectrum}

The following result follows directly from Proposition \ref{prop_str_con}. 
\begin{corollary} \label{corol_modified_str_con}
Let $\tilde{\mathcal{C}}^{(2N+1)}:=(\mathcal{V}^{(2N+1)})^{\frac{1}{2}}\mathcal{C}^{(2N+1)}(\mathcal{V}^{(2N+1)})^{\frac{1}{2}}$ and $\tilde{\mathcal{C}}:=\mathcal{V}^{\frac{1}{2}}\mathcal{C}\mathcal{V}^{\frac{1}{2}}$ be the generalised capacitance operators. Then, we have
\[
    \tilde{\mathcal{C}}^{(2N+1)} \to \tilde{\mathcal{C}} \text{ strongly,}\quad \text{as }N\to \infty.
\]
\end{corollary}
We briefly discuss how Corollary \ref{corol_modified_str_con} implies the convergence of the spectrum. The central problem is the following:
\begin{center}
    Given $\lambda_0\in \sigma(\tilde{\mathcal{C}})$, can $\lambda_0$ be approximated by the spectrum of $\tilde{\mathcal{C}}^{(2N+1)}$?
\end{center}
The answer is, as we shall see, positive. In particular, the argument is independent of the nature of $\lambda_0$, either in the discrete or continuous part of $\sigma(\tilde{\mathcal{C}})$. Indeed, since $\tilde{\mathcal{C}}$ and $\tilde{\mathcal{C}}^{(2N+1)}$ are self-adjoint, the strong convergence in Corollary \ref{corol_modified_str_con} implies that the strong resolvent convergence
\begin{equation} \label{eq_spec_con_1}
(\tilde{\mathcal{C}}^{(2N+1)}-\lambda)^{-1}\to (\tilde{\mathcal{C}}-\lambda)^{-1}\quad \text{strongly},\quad \forall \lambda\in U_b\cap\rho(\tilde{\mathcal{C}}),
\end{equation}
where $U_b:=\{\lambda\in\mathbb{C}:\, \sup_{N\geq 0}\|(\tilde{\mathcal{C}}^{(2N+1)}-\lambda)^{-1}\|_{\mathcal{B}(\lz)}<\infty\}$ and $\rho(\tilde{\mathcal{C}})$ is the resolvent set of $\tilde{\mathcal{C}}$; see \cite[Chapter VIII, Section 1, Corollary 1.6]{kato2013perturbation}, in which \eqref{eq_spec_con_1} is referred to as \textit{the generalised strong convergence}. Following directly by this resolvent convergence, we know that the spectrum of $\tilde{\mathcal{C}}^{(2N+1)}$ is lower semi-continuous in the following sense:
\begin{center}
For any open set $O\ni 
\lambda_0$, there exists $N_{O}>0$ such that $O\cap \sigma(\tilde{\mathcal{C}}^{(2N+1)})\neq \emptyset$ for all $N>N_{O}$.
\end{center}
Hence, choosing a sequence of neighbourhoods $(\lambda_0-1/k,\lambda_0+1/k)$, we conclude that the following result holds. 
\begin{corollary} \label{corol_lower_con_spectra}
Let $\lambda_0\in \sigma(\tilde{\mathcal{C}})$. There exists a subsequence $\tilde{\mathcal{C}}^{(2N_k+1)}$ whose spectrum approximates $\lambda_0$ in the sense that 
\begin{equation*}
\exists \lambda_k \in \sigma(\tilde{\mathcal{C}}^{(2N_k+1)}) \quad \text{s.t.}\quad \lambda_k\to \lambda_0 ,
\end{equation*}
as $k\to\infty$.
\end{corollary}
This recovers the convergence results for the defect modes and the essential spectrum that are proved in \cite{ammari.davies.ea2025Spectral, ammari.davies.ea2023Convergence}.

Moreover, when $\lambda_0$ belongs to the discrete spectrum of $\tilde{\mathcal{C}}$, we have the following stronger result on the convergence of the spectrum. Compared with Corollary \ref{corol_lower_con_spectra}, one should notice that the upper semi-continuity holds for the discrete spectrum, as indicated in \cite[Chapter VIII, Section 1.4]{kato2013perturbation}.

\begin{proposition} \label{prop_upp_con_spectrum}
Let $\sigma_{pp}(\tilde{\mathcal{C}})$ be the pure point spectrum of $\tilde{\mathcal{C}}$ and let $\lambda_0\in \sigma_{pp}(\tilde{\mathcal{C}})$. Then the spectrum of $\tilde{\mathcal{C}}^{(2N+1)}$ approximates $\lambda_0$ in the sense that there exist $\lambda_N \in \sigma(\tilde{\mathcal{C}}^{(2N+1)})$ and $C>0$ such that
\begin{equation*}
|\lambda_N-\lambda_0|<\frac{C}{\log^{\frac{1}{2}} N} ,
\end{equation*}
as $N\to\infty$.
\end{proposition}
The proof is based on the following Combes-Thomas-type estimate for discrete operators with polynomially off-diagonal decaying kernels.
\begin{lemma} \label{lem_inverse_closed}
Suppose that $z\in \rho(\tilde{\mathcal{C}})$. Then, the resolvent $\mathcal{R}_z=(\tilde{\mathcal{C}}-z)^{-1}$ admits the same off-diagonal decay estimate as $\tilde{\mathcal{C}}$, i.e., for $|n-m|$ being large
\begin{equation} \label{eq_inverse_closed_1}
\big|\mathcal{R}_z(n,m)\big|\leq \frac{C}{|n-m|\log^2|n-m|},
\end{equation}
where $C>0$ depends only on $\text{dist}(z,\sigma(\tilde{\mathcal{C}}))$. Consequently, the projection $P_{\lambda_0}=\mathbbm{1}_{\{\lambda_0\}}(\tilde{\mathcal{C}})$ satisfies
\begin{equation} \label{eq_inverse_closed_2}
\big|P_{\lambda_0}(n,m)\big|\leq \frac{C}{|n-m|\log^2|n-m|},
\end{equation}
for some $C>0$ and $|n-m|$ being large.
\end{lemma}
The proof of Lemma \ref{lem_inverse_closed} is based on the generalised Baskakov inverse-closedness theorem, presented in Appendix \ref{sec:app_D}. Now, we apply Lemma \ref{lem_inverse_closed} to prove Proposition \ref{prop_upp_con_spectrum}.

\begin{proof}[Proof of Proposition \ref{prop_upp_con_spectrum}]
{\color{blue}Step 1.} Since $\lambda_0\in \sigma_{pp}(\tilde{\mathcal{C}})$, we can select a normalised eigenfunction $u\in \ell^2(\mathbb{Z})$ associated with $\lambda_0$. Choosing $n_0\in \mathbb{Z}$ such that $u$ is peaked at $n_0$ and assume $n_0=0$ for simplicity, i.e., $|u(0)|=\|u\|_{\ell^{\infty}(\mathbb{Z})}$, which exists because $\ell^{2}(\mathbb{Z})\subset \ell^{\infty}(\mathbb{Z})$, the estimate \eqref{eq_inverse_closed_2} implies that $u_0$ decays as $|n|\to\infty$:
\begin{equation} \label{eq_upp_con_spectrum_proof_3}
|u(n)|=\frac{\big|P_{\lambda_0}(n,0)\big|}{|u(0)|}
\leq \frac{C\|u\|_{\ell^{\infty}(\mathbb{Z})}^{-1}}{|n|\log^2|n|},
\end{equation}
where in the first equality we have assumed without loss of generality that $P_{\lambda_0}$ is of rank one. We will use $u$ as a test function for the spectrum of the finite capacitance operator and show that, as $N\to\infty$,
\begin{equation} \label{eq_upp_con_spectrum_proof_1}
\|(\tilde{\mathcal{C}}^{(2N+1)}-\lambda_0)u\|_{\ell^2(\mathbb{Z})}\leq \frac{C}{\log^{\frac{1}{2}} N} ,
\end{equation}
by which it follows that
\begin{equation*}
\text{dist}\big(\sigma(\tilde{\mathcal{C}}^{(2N+1)}),\lambda_0\big)\leq \frac{C}{\log^{\frac{1}{2}} N},
\end{equation*}
and hence the proof of Proposition \ref{prop_upp_con_spectrum} is complete.

{\color{blue}Step 2.} To prove \eqref{eq_upp_con_spectrum_proof_1}, we introduce the truncation of $u$ defined as follows:
\begin{equation} \label{eq_cut_off}
u^{(N)}(n)=\left\{
\begin{aligned}
&u(n),\quad |n|<\frac{N}{2},\\
&0,\quad \text{otherwise.}
\end{aligned}
\right.
\end{equation}
Then the left side of \eqref{eq_upp_con_spectrum_proof_1} is decomposed as two parts
\begin{equation} \label{eq_upp_con_spectrum_proof_2}
\begin{aligned}
\|(\tilde{\mathcal{C}}^{(2N+1)}-\lambda_0)u\|_{\ell^2(\mathbb{Z})}
&\leq \|(\tilde{\mathcal{C}}^{(2N+1)}-\tilde{\mathcal{C}})u\|_{\ell^2(\mathbb{Z})} \\
&\leq \|(\tilde{\mathcal{C}}^{(2N+1)}-\tilde{\mathcal{C}})u^{(N)}\|_{\ell^2(\mathbb{Z})}
+\|(\tilde{\mathcal{C}}^{(2N+1)}-\tilde{\mathcal{C}})(u-u^{(N)})\|_{\ell^2(\mathbb{Z})} \\
&\leq \|(\tilde{\mathcal{C}}^{(2N+1)}-\tilde{\mathcal{C}})u^{(N)}\|_{\ell^2(\mathbb{Z})} \\
&\quad +\Big(\|\tilde{\mathcal{C}}^{(2N+1)} \|_{\mathcal{B}(\ell^2(\mathbb{Z}))}+\|\tilde{\mathcal{C}} \|_{\mathcal{B}(\ell^2(\mathbb{Z}))}\Big)\|u-u^{(N)}\|_{\ell^2(\mathbb{Z})} .
\end{aligned}
\end{equation}
Note that $\|\tilde{\mathcal{C}}^{(2N+1)} \|$ is uniformly bounded in $N$ as $\mathcal{C}^{(2N+1)}$ is
\begin{equation} \label{eq_upp_con_spectrum_proof_9}
\begin{aligned}
\|\mathcal{C}^{(2N+1)} \|_{\mathcal{B}(\ell^2(\mathbb{Z}))}
&\overset{(i)}{\leq} \sup_{n\in\mathbb{Z}}\sum_{m\in \mathbb{Z}}|\mathcal{C}^{(2N+1)}(n,m)| \\
&\overset{(ii)}{\leq} 2 \sup_{n\in\mathbb{Z}} \mathcal{C}^{(2N+1)}(n,n)
\overset{(iii)}{\leq} 2 \sup_{n\in\mathbb{Z}} \mathcal{C}(n,n)\overset{(iv)}{<}\infty,
\end{aligned}
\end{equation}
where the inequality (i) follows from the Schur test, (ii) is a consequence of the diagonal dominance of $\mathcal{C}^{(2N+1)}$, (iii) follows from the estimate \eqref{eq_str_con_proof_1} and (iv) results from the translation invariance of $\mathcal{C}$. On the other hand, the tail $u-u^{(N)}$ is estimated using \eqref{eq_upp_con_spectrum_proof_3}:
\begin{equation*}
\|u-u^{(N)}\|_{\ell^2(\mathbb{Z})}
\leq C\big(\sum_{|n|\geq \frac{N}{2}}\frac{1}{|n|^2\log^4|n|}\big)^{\frac{1}{2}}
\leq \frac{C}{N^{\frac{1}{2}}\log^{\frac{3}{2}}N} .
\end{equation*}
Hence, we conclude from \eqref{eq_upp_con_spectrum_proof_2} that
\begin{equation} \label{eq_upp_con_spectrum_proof_4}
\begin{aligned}
\|(\tilde{\mathcal{C}}^{(2N+1)}-\lambda_0)u\|_{\ell^2(\mathbb{Z})}
&\leq \|(\tilde{\mathcal{C}}^{(2N+1)}-\tilde{\mathcal{C}})u^{(N)}\|_{\ell^2(\mathbb{Z})}
+\frac{C}{N^{\frac{1}{2}}\log^{\frac{3}{2}}N} .
\end{aligned}
\end{equation}
The remaining paragraphs are devoted to estimate $\|(\tilde{\mathcal{C}}^{(2N+1)}-\tilde{\mathcal{C}})u^{(N)}\|_{\ell^2(\mathbb{Z})}$.

{\color{blue}Step 3.} By definition,
\begin{equation} \label{eq_upp_con_spectrum_proof_5}
\begin{aligned}
\|(\tilde{\mathcal{C}}^{(2N+1)}-\tilde{\mathcal{C}})u^{(N)}\|_{\ell^2(\mathbb{Z})}^2
=&\sum_{|n|\leq N}\big|\sum_{|m|<\frac{N}{2}}\big( \tilde{\mathcal{C}}^{(2N+1)}(n,m)-\tilde{\mathcal{C}}(n,m) \big)u(m)  \big|^2 \\
&+\sum_{|n|> N}\big|\sum_{|m|<\frac{N}{2}}\tilde{\mathcal{C}}(n,m) u(m)  \big|^2 .
\end{aligned}
\end{equation}
For the second sum, we apply the Cauchy-Schwarz inequality and see that
\begin{equation} \label{eq_upp_con_spectrum_proof_6}
\begin{aligned}
\sum_{|n|> N}\big|\sum_{|m|<\frac{N}{2}}\tilde{\mathcal{C}}(n,m) u(m)  \big|^2
&\leq \sum_{|n|> N}\sum_{|m|<\frac{N}{2}}|\tilde{\mathcal{C}}(n,m) |^2
\sum_{|m|<\frac{N}{2}}| u(m) |^2 \\
&\overset{(i)}{\leq} \sum_{|n|> N}\sum_{|m|<\frac{N}{2}}|\tilde{\mathcal{C}}(n,m) |^2 \\
&\overset{(ii)}{\leq} C\sum_{|n|> N}N\cdot \frac{1}{\big(|n|-\frac{N}{2}\big)^2\log^4\big(|n|-\frac{N}{2}\big)},
\end{aligned}
\end{equation}
where (i) follows from the normalisation $\|u\|_{\ell^2(\mathbb{Z})}=1$, and to derive the inequality (ii), we note that for $|n|>N,|m|<\frac{N}{2}$,
\begin{equation*}
|\tilde{\mathcal{C}}(n,m)|\leq \frac{C}{|n-m|\log^2|n-m|}
\leq \frac{C}{\big(|n|-\frac{N}{2}\big)\log^2\big(|n|-\frac{N}{2}\big)}.
\end{equation*}
The right side of \eqref{eq_upp_con_spectrum_proof_6} is estimated elementarily:
\begin{equation*}
\sum_{|n|> N}\frac{1}{\big(|n|-\frac{N}{2}\big)^2\log^4\big(|n|-\frac{N}{2}\big)}
\leq \frac{C}{N\log^3 N} .
\end{equation*}
Hence,
\begin{equation} \label{eq_upp_con_spectrum_proof_7}
\begin{aligned}
\sum_{|n|> N}\big|\sum_{|m|<\frac{N}{2}}\tilde{\mathcal{C}}(n,m) u(m)  \big|^2
\leq \frac{C}{\log^3 N}.
\end{aligned}
\end{equation}
For the first sum in \eqref{eq_upp_con_spectrum_proof_5}, we apply the Cauchy-Schwarz inequality again
\begin{equation} \label{eq_upp_con_spectrum_proof_8}
\begin{aligned}
&\sum_{|n|\leq N}\big|\sum_{|m|<\frac{N}{2}}\big( \tilde{\mathcal{C}}^{(2N+1)}(n,m)-\tilde{\mathcal{C}}(n,m) \big)u(m)  \big|^2 \\
&\leq \sum_{|n|\leq N}
\sum_{|m|<\frac{N}{2}}\big| \tilde{\mathcal{C}}^{(2N+1)}(n,m)-\tilde{\mathcal{C}}(n,m) \big|^2 |u(m)|\sum_{|m|<\frac{N}{2}}|u(m)| \\
&\overset{(i)}{\leq} C \sum_{|n|\leq N}
\sum_{|m|<\frac{N}{2}}\big| \tilde{\mathcal{C}}^{(2N+1)}(n,m)-\tilde{\mathcal{C}}(n,m) \big|^2 \frac{1}{|m|\log^2|m|} \\
&\leq C\sum_{|m|<\frac{N}{2}}\frac{1}{|m|\log^2|m|}
\sup_{|n|\leq N}\big| \tilde{\mathcal{C}}^{(2N+1)}(n,m)-\tilde{\mathcal{C}}(n,m) \big| \sum_{|n|\leq N}\big| \tilde{\mathcal{C}}^{(2N+1)}(n,m)-\tilde{\mathcal{C}}(n,m) \big| , 
\end{aligned}
\end{equation}
where we have applied \eqref{eq_upp_con_spectrum_proof_3} to derive (i). We now give an $(m,N)$-independent estimate for the two terms involving $\tilde{\mathcal{C}}^{(2N+1)}(n,m)-\tilde{\mathcal{C}}(n,m)$. First, arguing similarly as in \eqref{eq_upp_con_spectrum_proof_9}, we have
\begin{equation} \label{eq_upp_con_spectrum_proof_10}
\begin{aligned}
\sum_{|n|\leq N}\big| \tilde{\mathcal{C}}^{(2N+1)}(n,m)-\tilde{\mathcal{C}}(n,m) \big| 
&\leq \sum_{|n|\leq N} | \tilde{\mathcal{C}}^{(2N+1)}(n,m) |+ |\tilde{\mathcal{C}}(n,m) | \\
&\leq \sup_{m\in\mathbb{Z}}\Big[\mathcal{C}(m,m)+\sum_{n\in\mathbb{Z}}|\mathcal{C}(n,m)|\Big] <\infty .
\end{aligned}
\end{equation}
On the other hand, using estimates \eqref{eq_str_con_proof_7} and \eqref{eq_str_con_proof_6}, we have
\begin{equation} \label{eq_upp_con_spectrum_proof_11}
\begin{aligned}
\sup_{|n|\leq N}\big| \tilde{\mathcal{C}}^{(2N+1)}(n,m)-\tilde{\mathcal{C}}(n,m) \big| 
&\leq \sup_{|n|\leq N} \sum_{|j|\geq N} |\mathcal{C}(j,m)|
=\sum_{|j|\geq N} |\mathcal{C}(j,m)| \\
&\leq \sum_{|j|\geq N}\frac{C}{\big(|j|-\frac{N}{2}\big)\log^2\big(|j|-\frac{N}{2}\big)} \leq \frac{C}{\log N} .
\end{aligned}
\end{equation}
Thus, by \eqref{eq_upp_con_spectrum_proof_8}-\eqref{eq_upp_con_spectrum_proof_11}, we conclude
\begin{equation} \label{eq_upp_con_spectrum_proof_12}
\sum_{|n|\leq N}\big|\sum_{|m|<\frac{N}{2}}\big( \tilde{\mathcal{C}}^{(2N+1)}(n,m)-\tilde{\mathcal{C}}(n,m) \big)u(m)  \big|^2
\leq \frac{C}{\log N}\sum_{|m|<\frac{N}{2}}\frac{1}{|m|\log^2|m|}
\leq \frac{C}{\log N} .
\end{equation}
Then the proof of \eqref{eq_upp_con_spectrum_proof_1} is complete by \eqref{eq_upp_con_spectrum_proof_4}, \eqref{eq_upp_con_spectrum_proof_5}, \eqref{eq_upp_con_spectrum_proof_7} and \eqref{eq_upp_con_spectrum_proof_12}.
\end{proof}

\begin{remark} \label{rmk_optimality}
We note that, the logarithmic convergence rate proved in Proposition \ref{prop_upp_con_spectrum} is far from being optimal. In fact, in \cite{ammari.davies.ea2023Convergence}, 
it was  observed numerically that the actual convergence rate of the discrete spectrum is polynomial. This problem can be resolved by improving the estimate of the left side of \eqref{eq_upp_con_spectrum_proof_8}. Indeed, by adjusting the cut-off scaling (e.g., setting $N^p$ with $p<1$ instead of $N/2$ in \eqref{eq_cut_off}), the estimate \eqref{eq_upp_con_spectrum_proof_7} will become algebraic, but \eqref{eq_upp_con_spectrum_proof_12} is still logarithmic. The only way to improve this logarithmic bound is to replace \eqref{eq_upp_con_spectrum_proof_11} with a sharper estimate. This, in turn, requires a better estimate of the difference between the capacitance elements than the one presented in \eqref{eq_str_con_proof_3}. The problem is very challenging due to the borderline decay of the off diagonal entries of the capacitance operator. We leave this problem for future work.
\end{remark}

\begin{remark}
 Note  that our capacitance operator formalism can also be generalised to infinite structures that are not translationally invariant and have interfaces such as the analogue of the well-known interface Su-Schrieffer-Heeger model, and to estimate the rate of convergence of the localised interface mode in the truncated structure to the one in the infinite structure, as numerically illustrated in \cite{ammari.davies.ea2025Spectral}. 
\end{remark}

\begin{remark} Note that for crystals of high-contrast resonators, i.e. 3D-in-3D, as proved in \cite{ammari.qiu2025Analysis}, the off-diagonal entries of the capacitance operator decay exponentially. In this case, it is easy to prove that for any $\lambda_0 \in \sigma_{pp}(\tilde{\mathcal{C}})$, there exists a sequence $\lambda_N \in \sigma(\tilde{\mathcal{C}}^{(2N+1)})$ that convergences exponentially to $\lambda_0$.

\end{remark}

\section{Anderson localisation for arbitrary disorder} \label{sec:anderson}
Having established the long-range interaction and the strong convergence of the finite generalised capacitance matrix $\mc V^{(2N+1)} \mc C^{(2N+1)}$ to the generalised capacitance operator $\mathcal{V} \mc{C}$ on $\ell^2(\Z)$, we now aim to investigate its spectral properties as a disorder is introduced in the form of a material parameter perturbation. In this section, we shall consider an infinite chain of spherical resonators and randomly perturb the material parameters by choosing 
\[
    v_i \sim 1 + \mathcal{U}([-\varepsilon, \varepsilon]),
\]
where $\mathcal{U}([-\varepsilon, \varepsilon])$ is a uniform distribution with 
zero mean-value and standard variation $\varepsilon/\sqrt{3}$, yielding a random material parameter operator
\begin{equation}
    \begin{aligned}
        \mc{V}:\ell^2(\Z) &\to\ell^2(\Z) \\
        (\bseq{u}{i})_{i\in \Z} &\mapsto \left(\frac{v_i^2}{\abs{D_i}}\bseq{u}{i}\right)_{i\in \Z}.
    \end{aligned}
\end{equation}
For the case of identical spherical resonators of radius $R$, we have $\abs{D_i} = \frac{4}{3}\pi R^3$ for all $i\in \Z$. In line with \cref{eq:simprob}, the infinite periodic system is then characterised by the generalised capacitance operator 
\[
    \tilde{\mathcal{C}} \coloneqq \mc V^{\frac{1}{2}} \mc C \mc V^{\frac{1}{2}}.
\]
The central aim of this section is to demonstrate numerically \emph{global localisation for arbitrary disorder}. To be precise, this means that the infinite operator $\tilde{\mathcal{C}}$ has only (apart from a set of measure zero) pure point (i.e., $\ell^2$-summable) spectrum for any amount of material parameter perturbation $0<\varepsilon<1$. This expectation is based on (i) the one-dimensional geometry of our system together with (ii) the summability of the off-diagonal interactions established in the previous section. It is well-known from \cite{ishii1973Localization} that tight-binding one-dimensional systems exhibit global localisation for global disorder. Moreover, this behaviour is expected to extend to one-dimensional systems that exhibit long-range interactions, as long as they remain summable \cite{deng2018Duality}. The key difference between tight-binding and long-range interactions is the weaker decay of the localised modes (typically algebraic instead of exponential).

Clearly, the strong convergence $\mc{C}^{(2N+1)} \to \mc C$ induces the strong convergence 
\[
    (\mc{V}^{(2N+1)})^{\frac{1}{2}} \mc{C}^{(2N+1)}  (\mc{V}^{(2N+1)})^{\frac{1}{2}} \to \mc V^{\frac{1}{2}}\mc C\mc V^{\frac{1}{2}}.
\]
To understand the localisation behaviour of the infinite system, we may thus investigate the convergence of the finite counterparts.
A key diagnostic is the \emph{inverse participation ratio}
\[
    \on{IPR}(\bm u) \coloneqq \frac{\norm{\bm u}_4}{\norm{\bm u}_2}.
\]

\begin{figure}
    \centering
    \begin{subfigure}[t]{0.49\textwidth}
        \centering
        \includegraphics[width=\linewidth]{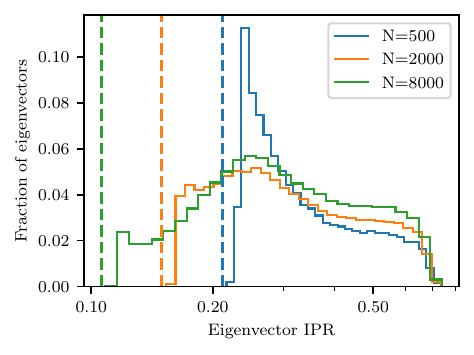}
        \caption{IPR distribution of systems with randomly perturbed ($\varepsilon=0.03$, averaged over $100$ realisations each) material parameters. The lowest possible IPR ($N^{-1/4}$) for a given system size $N$ is indicated by a vertical dashed line.  As the system size is increased, increasingly weakly localised modes may be resolved.}
        \label{fig:IPR_distribution}
    \end{subfigure}\hfill
    \begin{subfigure}[t]{0.49\textwidth}
        \centering
        \includegraphics[width=\linewidth]{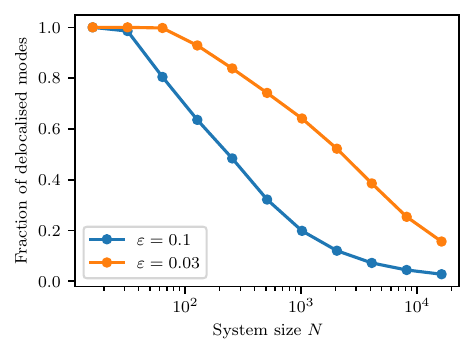}
        \caption{Fraction of delocalised modes (identified by $\mathrm{IPR}<2N^{-1/4}$) for increasing system size $N$. We investigate two different noise levels and average over $100$ realisations to find the mean and standard deviation. In both cases, the fraction of delocalised modes approaches zero, but it does so significantly faster at higher noise levels.}
        \label{fig:Fraction_delocalised}
    \end{subfigure}%
    \caption{IPR statistics of randomly perturbed chains as the number of resonators in increased. The fraction of eigenmodes with IPR close to the minimal $N^{-1/4}$ decreases with increasing system size, suggesting global localisation in the infinite limit.}
\end{figure}

For a finite vector $\bm u\in \R^N$ we have the dichotomy
\[
    \on{IPR}(\bm u) = \begin{cases}
        \mc O(1) & \bm u \text{ localised / }\ell^2\text{-summable},\\
        \mc O(N^{-1/4}) & \bm u\text{ delocalised}.
    \end{cases}
\]
stemming from the two extreme cases $\bseq{u}{i} = \delta_{0i}$ and $\bseq{u}{i} = N^{-1/2}$. This implies that the lowest degree of localisation that can be resolved by a system of size $N$ is $N^{-1/4}$. Indeed, in \cref{fig:IPR_distribution}, we see that the low end of the IPR distribution is limited by the size of the system. As the system size increases, larger localisation lengths may be resolved. 

Consequently, to distinguish delocalisation from localisation at a large scale, it is necessary to investigate the IPR distribution as the system size $N$ increases. In this case, we can determine that all eigenmodes of the infinite system are localised if and only if the fraction of delocalised modes (i.e. IPR of order $\mc O(N^{-1/4})$) goes to $0$ as $N\to \infty$. In practice, we determine a mode to be delocalised if its IPR lies below some fixed multiple of $\mc O(N^{-1/4})$. In \cref{fig:Fraction_delocalised}, we plot the fraction of delocalised modes at different noise strengths $\varepsilon$ as $N$ increases. In both cases, the fraction of delocalised modes approaches $0$ with increasing system size, indicating global localisation.

Finally, in this section, we aim to investigate the key difference introduced by long-range interactions. Namely, the slow decay of the interactions limits the maximal eigenmode decay, preventing the existence of exponentially localised eigenmodes.
As can be seen in \cref{fig:mean_decay}, the eigenmode decay follows a two-scale behaviour: close to the localisation site, the decay behaviour is exponential, while further away the maximal decay is bounded from below by an algebraic envelope.

\begin{figure}
    \centering
    \includegraphics[width=0.5\linewidth]{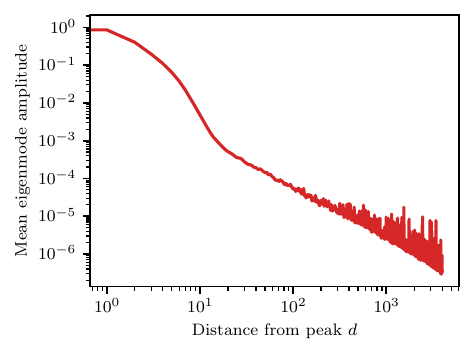}
    \caption{Mean eigenmode amplitude $d$ sites away from the eigenmode peak. Averaged over the eigenmodes corresponding to the $M=200$ largest eigenvalues of a perturbed system ($\varepsilon=0.1, N=8000$, averaged over 1000 realisations). Away from their localisation site, the eigenmodes decay algebraically.}
    \label{fig:mean_decay}
\end{figure}

\section{Concluding remarks}
In this paper, we have proved for the first time the algebraic decay of the off-diagonal entries of the capacitance operator of an infinite chain of high-contrast resonators. We have also proved the strong convergence of the capacitance matrix of the truncated chain to the corresponding capacitance operator of the infinite chain as the size of the truncation goes to infinity. Moreover, our result on the convergence of the spectrum has been refined to obtain an estimate of the convergence of defect modes in a finite chain to those of the corresponding infinite chain as the size of the chain goes to infinity. We have also shown that this decay rate of the off-diagonal entries is a consequence of long-range interactions between the resonators and is borderline for Anderson localisation to occur for arbitrary disorder. We have numerically illustrated the localisation of the eigenmodes for arbitrary fluctuations in the material parameter of the resonators and their algebraic decay from their localisation sites. 

It would be very interesting to generalise and contrast the results of this paper on the algebraic decay of the off-diagonal capacitance operator and the occurrence of Anderson localisation in screens of high-contrast resonators, i.e., two-dimensional periodic structures in three dimensions. In that setting, Anderson localisation for arbitrary fluctuations in the material parameters of the resonators is expected to hold with again an algebraic rate of decay (but faster than the one for chains) for the eigenmodes of the capacitance operator from their sites. This would be the subject of future work. Moreover, for crystals of high-contrast resonators, i.e., three-dimensional periodic structures in three dimensions, although the off-diagonal entries of the associated capacitance operator decay exponentially and, consequently, a tight-binding approximation applies \cite{ammari.qiu2025Analysis}, it is expected that a phase transition in localisation occurs due to percolation. At a critical disorder strength, eigenmodes switch from delocalised to localised. In future work, we plan to simulate such three-dimensional computationally heavy systems to illustrate the phase transition in localisation and to estimate the mobility edge that separates extended from localised eigenmodes in these classical three-dimensional wave systems in terms of their geometries.

\section*{Acknowledgments}
    This work was supported by the Swiss National Science Foundation grant number 200021-236472 (A.U. and J.Q.) and 235080 (S.B.).
    
\section*{Code availability}
The software used to produce the numerical results in this work is openly available at \\ \href{https://doi.org/10.5281/zenodo.21002014}{https://doi.org/10.5281/zenodo.21002014}.

\appendix

\section{Proofs of Propositions \ref{prop_cap_def_momentum_existence_proof_1} and \ref{prop_cap_def_momentum_existence_proof_2}}
\label{sec:app_A}

\begin{proof}[Proof of Proposition \ref{prop_cap_def_momentum_existence_proof_1}]
The quasi-periodicity \eqref{eq_cap_def_momentum_existence_proof_4} follows directly from the definition. To prove the remaining properties, we introduce the following decomposition:
\begin{equation} \label{eq_app_A_3}
\begin{aligned}
K^{\alpha}(\bm{x},\bm{y})&=\sum_{n\geq 1} \big[\frac{1}{|\bm{x}-\bm{y}-n\bm{e}_1|}-\frac{1}{|n|}\big]e^{i\alpha n}+\sum_{n\leq -1} \big[\frac{1}{|\bm{x}-\bm{y}-n\bm{e}_1|}-\frac{1}{|n|}\big]e^{i\alpha n}+\sum_{n\neq 0}\frac{e^{i\alpha n}}{|n|} \\
&=:K^{\alpha}_{+}(\bm{x},\bm{y})+K^{\alpha}_{-}(\bm{x},\bm{y})-2\log\big(2\sin \frac{|\alpha|}{2}\big) .
\end{aligned}
\end{equation}
Since $\log\big(2\sin \frac{|\alpha|}{2}\big)$ is smooth in $\alpha$ within $[-\pi,0)\cup(0,\pi]$, it suffices to prove the conclusion of Proposition \ref{prop_cap_def_momentum_existence_proof_1} for $K^{\alpha}_{\pm}(\bm{x},\bm{y})$. We only show the proof for $K^{\alpha}_{+}(\bm{x},\bm{y})$, and $K^{\alpha}_{-}(\bm{x},\bm{y})$ is treated similarly. Fix a compact set $V$ such that $D\Subset V\Subset Y$. Note that $\gamma:=\sup_{\bm{x},\bm{y}\in V}|x_1-y_1|<1$. Hence, using the Laplace transform, we obtain
\begin{equation*}
\begin{aligned}
K^{\alpha}_{+}(\bm{x},\bm{y})=\sum_{n\geq 1} \big[\frac{1}{\sqrt{|z_{\perp}|^2+(z_1-n)^2}} -\frac{1}{n}\big]e^{i\alpha n}
=\sum_{n\geq 1}\int_{0}^{\infty}\big(J_0(|z_{\perp}t|)e^{z_1 t}-1\big)e^{-nt+in\alpha}dt
\end{aligned}
\end{equation*}
with $\bm{z}:=\bm{x}-\bm{y}=(z_1,z_{\perp})$ and $J_0$ being the Bessel function. Note that, since $z_1\leq \gamma<1$ and $J_0(|z_{\perp}t|)e^{z_1 t}-1=\mathcal{O}(t)$ as $t\to 0$, the right side of the above integral is absolutely integrable:
\begin{equation} \label{eq_app_A_1}
\begin{aligned}
\sum_{n\geq 1}\int_{0}^{\infty}\Big|\big(J_0(|z_{\perp}t|)e^{z_1 t}-1\big)e^{-nt+in\alpha}\Big|dt
&=\int_{0}^{\infty}\frac{\big|J_0(|z_{\perp}t|)e^{z_1 t}-1 \big|}{1-e^{-t}}e^{-t}dt \lesssim \int_{0}^{\infty}e^{-(1-\gamma)t}dt <\infty.
\end{aligned}
\end{equation}
Here, $A \lesssim B$ means that there exists a constant $C$ such that $A \leq C B$.
Thus, we can interchange the summation and integral and obtain
\begin{equation} \label{eq_app_A_2}
K^{\alpha}_{+}(\bm{x},\bm{y})
=\int_{0}^{\infty}\big(J_0(|z_{\perp}t|)e^{z_1 t}-1\big)\sum_{n\geq 1}e^{-nt+in\alpha}dt = \int_{0}^{\infty}\big(J_0(|z_{\perp}t|)e^{z_1 t}-1\big)\frac{e^{-t+i\alpha}}{1-e^{-t+i\alpha}}dt .
\end{equation}
For any $\alpha\in [-\pi,\pi]$, the uniform boundedness of Bessel functions shows that
\begin{equation} \label{eq_app_A_9}
\int_{0}^{\infty}\Big|\partial_{z}^{\beta}\big(J_0(|z_{\perp}t|)e^{z_1 t}-1\big)\frac{e^{-t+i\alpha}}{1-e^{-t+i\alpha}}\Big|dt \lesssim \int_{0}^{\infty}\frac{t}{|1-e^{-t+i\alpha}|}e^{-(1-\gamma)t}<\infty,
\end{equation}
for any multi-index $\beta$. This concludes that $K^{\alpha}_{+}\in C^{1}(\partial D\times \partial D)$. Similarly, one can prove \eqref{eq_app_A_2} is also absolutely integrable by taking any order of $\alpha$-derivatives for $\alpha\neq 0$; the details are skipped. This indicates that $K^{\alpha}_{+}$ is smooth in $\alpha$ and concludes the proof.
\end{proof}

\begin{proof}[Proof of Proposition \ref{prop_cap_def_momentum_existence_proof_2}]
Corresponding to the decomposition \eqref{eq_cap_def_momentum_existence_proof_1} and \eqref{eq_app_A_3}, we write
\begin{equation} \label{eq_app_A_8}
\begin{aligned}
u(\bm{x})&=-\frac{1}{4\pi}\int_{\partial D}d\sigma_{\bm{y}}\varphi(\bm{y})\frac{1}{|\bm{x}-\bm{y}|}  -\frac{1}{2\pi}\log(2\sin\frac{|\alpha|}{2})\int_{\partial D}d\sigma_{\bm{y}}\varphi(\bm{y}) \\
&\quad -\frac{1}{4\pi}\int_{\partial D}d\sigma_{\bm{y}}\varphi(\bm{y})\Big[\sum_{n\geq 1} \big[\frac{1}{|\bm{x}-\bm{y}-n\bm{e}_1|}-\frac{1}{|n|}\big]e^{i\alpha n}\Big] \\
&\quad -\frac{1}{4\pi}\int_{\partial D}d\sigma_{\bm{y}}\varphi(\bm{y})\Big[\sum_{n\leq -1} \big[\frac{1}{|\bm{x}-\bm{y}-n\bm{e}_1|}-\frac{1}{|n|}\big]e^{i\alpha n}\Big] \\
&=-\frac{1}{4\pi}\int_{\partial D}d\sigma_{\bm{y}}\varphi(\bm{y})\frac{1}{|\bm{x}-\bm{y}|}  -\frac{1}{2\pi}\log(2\sin\frac{|\alpha|}{2})\int_{\partial D}d\sigma_{\bm{y}}\varphi(\bm{y}) \\
&\quad -\frac{1}{4\pi}\int_{\partial D}d\sigma_{\bm{y}}\varphi(\bm{y})\int_{0}^{\infty}\big(J_0(|x_{\perp}-y_{\perp}|t)e^{(x_1-y_1) t}-1\big)\frac{e^{-t+i\alpha}}{1-e^{-t+i\alpha}}dt \\
&\quad -\frac{1}{4\pi}\int_{\partial D}d\sigma_{\bm{y}}\varphi(\bm{y})\int_{0}^{\infty}\big(J_0(|x_{\perp}-y_{\perp}|t)e^{-(x_1-y_1) t}-1\big)\frac{e^{-t-i\alpha}}{1-e^{-t-i\alpha}}dt \\
&=-\frac{1}{4\pi}\int_{\partial D}d\sigma_{\bm{y}}\varphi(\bm{y})\frac{1}{|\bm{x}-\bm{y}|}  \\
&\quad -\frac{1}{4\pi}\int_{\partial D}d\sigma_{\bm{y}}\varphi(\bm{y})\int_{0}^{\infty}J_0(|x_{\perp}-y_{\perp}|t)e^{(x_1-y_1) t}\frac{e^{-t+i\alpha}}{1-e^{-t+i\alpha}}dt \\
&\quad -\frac{1}{4\pi}\int_{\partial D}d\sigma_{\bm{y}}\varphi(\bm{y})\int_{0}^{\infty}J_0(|x_{\perp}-y_{\perp}|t)e^{-(x_1-y_1) t}\frac{e^{-t-i\alpha}}{1-e^{-t-i\alpha}}dt \\
&=: u^{free}(\bm{x})+u^{rem,\alpha}_{+}(\bm{x})+u^{rem,\alpha}_{-}(\bm{x}) ,
\end{aligned}
\end{equation}
where the separation of the integral is legal for $\alpha\neq 0$. The function $u^{free}(\bm{x})$ is the single-layer potential associated with the free-space Green function, which is harmonic within $Y\backslash \overline{D}$ and $D$, satisfies the jump identities \eqref{eq_cap_def_momentum_existence_proof_2} and the asymptotics \eqref{eq_cap_def_momentum_existence_proof_3}. Let us now consider the remaining parts $u^{rem,\alpha}_{\pm}$. As seen from the proof of Proposition \ref{prop_cap_def_momentum_existence_proof_1}, $u^{rem,\alpha}_{\pm}(\bm{x})$ are smooth in a neighborhood of $\partial D$, which implies that $u(\bm{x})$ satisfies the jump identities \eqref{eq_cap_def_momentum_existence_proof_2}. On the other hand, as justified in Proposition \ref{prop_cap_def_momentum_existence_proof_1}, we can plug the spatial derivatives into the summation defining $u^{rem,\alpha}_{\pm}$
\begin{equation*}
\Delta u^{rem,\alpha}_{+}(\bm{x})=-\frac{1}{4\pi}\int_{\partial D}d\sigma_{\bm{y}}\varphi(\bm{y})\int_{0}^{\infty}\Delta_{\bm{x}}\big[J_0(|x_{\perp}-y_{\perp}|t)e^{(x_1-y_1) t}\big]\frac{e^{-t+i\alpha}}{1-e^{-t+i\alpha}}dt=0
\end{equation*}
for $\bm{x}\notin \partial D$. Hence, the harmonicity of $u(\bm{x})$ is also proved. Now we are left to prove that $u^{rem,\alpha}_{\pm}(\bm{x})$ satisfies the asymptotics \eqref{eq_cap_def_momentum_existence_proof_3}. We will only show the proof for $u^{rem,\alpha}_{+}(\bm{x})$, while $u^{rem,\alpha}_{-}(\bm{x})$ is treated similarly. Let $x_{\perp}$ be sufficiently large such that
\begin{equation} \label{eq_eq_app_A_3}
|x_{\perp}|>\max\big\{1,\sup_{\bm{y}\in\partial D}|y_{\perp}|\big\} .
\end{equation}
Then, we decompose
\begin{equation} \label{eq_eq_app_A_4}
\begin{aligned}
&\int_{0}^{\infty}J_0(|x_{\perp}-y_{\perp}|t)e^{(x_1-y_1) t}\frac{e^{-t+i\alpha}}{1-e^{-t+i\alpha}}dt \\
&=\int_{0}^{\frac{1}{|x_{\perp}|^p}}J_0(|x_{\perp}-y_{\perp}|t)e^{(x_1-y_1) t}\frac{e^{-t+i\alpha}}{1-e^{-t+i\alpha}}dt \\
&+ \int_{\frac{1}{|x_{\perp}|^p}}^{\infty}J_0(|x_{\perp}-y_{\perp}|t)e^{(x_1-y_1) t}\frac{e^{-t+i\alpha}}{1-e^{-t+i\alpha}}dt, 
\end{aligned}
\end{equation}
where $p\in(0,1)$ is a constant to be determined. For $\alpha\neq 0$, the integrand is uniformly bounded for $\bm{x}\in Y,\bm{y}\in\partial D$. Hence,
\begin{equation} \label{eq_eq_app_A_5}
\int_{0}^{\frac{1}{|x_{\perp}|^p}}J_0(|x_{\perp}-y_{\perp}|t)e^{(x_1-y_1) t}\frac{e^{-t+i\alpha}}{1-e^{-t+i\alpha}}dt=\mathcal{O}(|x_{\perp}|^{-p}) .
\end{equation}
For the other integral in \eqref{eq_eq_app_A_4}, we apply the asymptotics of Bessel functions and see that
\begin{equation} \label{eq_eq_app_A_6}
\begin{aligned}
&\int_{\frac{1}{|x_{\perp}|^p}}^{\infty}J_0(|x_{\perp}-y_{\perp}|t)e^{(x_1-y_1) t}\frac{e^{-t+i\alpha}}{1-e^{-t+i\alpha}}dt \\
&=\int_{\frac{1}{|x_{\perp}|^p}}^{\infty}\sqrt{\frac{2}{\pi|x_{\perp}-y_{\perp}|t}}\cos(|x_{\perp}-y_{\perp}|t-\frac{\pi}{4})e^{(x_1-y_1) t}\frac{e^{-t+i\alpha}}{1-e^{-t+i\alpha}}dt \\
&\quad +\int_{\frac{1}{|x_{\perp}|^p}}^{\infty} \mathcal{O}\big(\frac{1}{|x_{\perp}-y_{\perp}|^{3/2}t^{3/2}}\big)\cdot e^{(x_1-y_1) t}\frac{e^{-t+i\alpha}}{1-e^{-t+i\alpha}}dt \\
&=\int_{\frac{1}{|x_{\perp}|^p}}^{\infty}\sqrt{\frac{2}{\pi|x_{\perp}-y_{\perp}|t}}\cos(|x_{\perp}-y_{\perp}|t-\frac{\pi}{4})e^{(x_1-y_1) t}\frac{e^{-t+i\alpha}}{1-e^{-t+i\alpha}}dt + \mathcal{O}\big(|x_{\perp}|^{-\frac{3}{2}(1-p)}\big).
\end{aligned}
\end{equation}
An integration by parts gives
\begin{equation} \label{eq_eq_app_A_7}
\begin{aligned}
&\int_{\frac{1}{|x_{\perp}|^p}}^{\infty}\sqrt{\frac{2}{\pi|x_{\perp}-y_{\perp}|t}}\cos(|x_{\perp}-y_{\perp}|t-\frac{\pi}{4})e^{(x_1-y_1) t}\frac{e^{-t+i\alpha}}{1-e^{-t+i\alpha}}dt \\
&=\sqrt{\frac{2}{\pi|x_{\perp}-y_{\perp}|^{3}|x_{\perp}|^{-p}}}\sin(|x_{\perp}-y_{\perp}|t-\frac{\pi}{4})e^{(x_1-y_1) t}\frac{e^{-t+i\alpha}}{1-e^{-t+i\alpha}}\Big|_{t=1/|x_{\perp}|^p} \\
&\quad +\frac{1}{2}\int_{\frac{1}{|x_{\perp}|^p}}^{\infty}\sqrt{\frac{2}{\pi|x_{\perp}-y_{\perp}|^3 t^3}}\cos(|x_{\perp}-y_{\perp}|t-\frac{\pi}{4})e^{(x_1-y_1) t}\frac{e^{-t+i\alpha}}{1-e^{-t+i\alpha}}dt \\
&\quad -\int_{\frac{1}{|x_{\perp}|^p}}^{\infty}\sqrt{\frac{2}{\pi|x_{\perp}-y_{\perp}|^3 t}}\cos(|x_{\perp}-y_{\perp}|t-\frac{\pi}{4})\partial_{t}\big[e^{(x_1-y_1) t}\frac{e^{-t+i\alpha}}{1-e^{-t+i\alpha}}\big]dt \\
&=\mathcal{O}\big(|x_{\perp}|^{-\frac{3}{2}(1-p)}+|x_{\perp}|^{-\frac{3}{2}(1-p)}+|x_{\perp}|^{-\frac{3}{2}(1-\frac{1}{3}p)}\big) =\mathcal{O}(|x_{\perp}|^{-\frac{3}{2}(1-p)}) .
\end{aligned}
\end{equation}
By \eqref{eq_eq_app_A_4}-\eqref{eq_eq_app_A_7}, we select $p\in (\frac{1}{2},\frac{2}{3})$ and see that
\begin{equation*}
u^{rem,\alpha}_{+}(\bm{x})=o\big(|x_{\perp}|^{-\frac{1}{2}}\big) .
\end{equation*}
The estimate of $\nabla_{\perp}u^{rem,\alpha}_{+}(\bm{x})$ is similar, for which one replaces $J_0(|x_{\perp}-y_{\perp}|t)$ in \eqref{eq_eq_app_A_4}-\eqref{eq_eq_app_A_7} by $\nabla_{\perp}J_0(|x_{\perp}-y_{\perp}|t)\sim tJ_1(|x_{\perp}-y_{\perp}|t)$ and then follows the same lines; the details are skipped here.
\end{proof}

\section{Proofs of Propositions \ref{prop_remainder_R_alpha} and \ref{prop_A_alpha_invertibility} }
\label{sec:app_B}

\begin{proof}[Proof of Proposition \ref{prop_remainder_R_alpha}]
We only prove for the case of $\alpha>0$ and the sum over positive indices, i.e.,
\begin{equation*}
\mathcal{R}_{+}^{\alpha}[\varphi](\bm{x}):=\int_{\partial D}d\sigma_{\bm{y}}\varphi(\bm{y})\Big[\sum_{n\geq 1} \big(\frac{1}{|\bm{x}-\bm{y}-n\bm{e}_1|}-\frac{1}{|n\bm{e}_1|}\big)(e^{i\alpha n}-1)\Big],
\end{equation*}
while the negative sum is treated similarly. Following the same argument as those in the proof of Proposition \ref{prop_cap_def_momentum_existence_proof_1}, we rewrite the infinite sum using the Laplace transform
\begin{equation} \label{eq_app_B_2}
\mathcal{R}_{+}^{\alpha}[\varphi](\bm{x})=\int_{\partial D}d\sigma_{\bm{y}}\varphi(\bm{y})\int_{0}^{\infty}H(\bm{z};t)
\big(\frac{e^{-t+i\alpha}}{1-e^{-t+i\alpha}}-\frac{e^{-t}}{1-e^{-t}}\big)dt
\end{equation}
with $\bm{z}=\bm{x}-\bm{y}$, $H(\bm{z};t):=J_0(|z_{\perp}t|)e^{z_1 t}-1$ and $J_0$ being the Bessel function. Since 
\begin{equation} \label{eq_app_B_3}
H(\bm{z};t)=\mathcal{O}(e^{\gamma t})\text{ as }t\to\infty,\quad H(\bm{z};t)=\mathcal{O}(t)\text{ as }t\to 0,
\end{equation}
where $\gamma\in(0,1)$ is introduced in the proof of Proposition \ref{prop_cap_def_momentum_existence_proof_1}, one sees that the $\alpha$-derivatives (of any order) of the right side of \eqref{eq_app_B_2} is absolutely integrable for $\alpha\neq 0$, which proves the smoothness of $\mathcal{R}_{+}^{\alpha}$. 

Now, we estimate $\|\partial_{\alpha}^{n}\mathcal{R}_{+}^{\alpha}\|_{\mathcal{B}(L^2(\partial D),H^1(\partial D))}$. Since $\nabla_{\bm{z}}H(\bm{z};t)$ satisfies the same asymptotics \eqref{eq_app_B_3}, we only need to estimate an elementary integral:
\begin{equation} \label{eq_app_B_4}
\begin{aligned}
\|\partial_{\alpha}^{n}\mathcal{R}_{+}^{\alpha}\|_{\mathcal{B}(L^2(\partial D),H^1(\partial D))} &\lesssim \int_{0}^{\infty}\big|\partial_{\alpha}^{n}\big(\frac{(1-e^{-t})e^{i\alpha}}{1-e^{-t+i\alpha}}-1\big)\big|e^{-(1-\gamma)t}dt \\
&=\int_{0}^{\infty}\big|\partial_{\alpha}^{n}\frac{e^{i\alpha}-1}{1-e^{-t+i\alpha}}\big|e^{-(1-\gamma)t}dt .
\end{aligned}
\end{equation}
The estimate of $n=0$ proceeds as follows:
\begin{equation*}
\begin{aligned}
\int_{0}^{\infty}\big|\frac{e^{i\alpha}-1}{1-e^{-t+i\alpha}}\big|e^{-(1-\gamma)t}dt
&=\mathcal{O}(\alpha)\int_{0}^{\infty}\big|\frac{1}{1-e^{-t+i\alpha}}\big|e^{-(1-\gamma)t}dt \\
&=\mathcal{O}(\alpha)\int_{0}^{1}\big|\frac{1}{1-e^{-t+i\alpha}}\big|e^{-(1-\gamma)t}dt +\mathcal{O}(\alpha) \\
&\leq \mathcal{O}(\alpha)\int_{0}^{1}\frac{1}{1-\cos{\alpha}+\frac{\cos{\alpha}}{2}t}dt +\mathcal{O}(\alpha) \\
&=\mathcal{O}(\alpha\log\alpha).
\end{aligned}
\end{equation*}
The higher-order derivatives are estimated similarly.
\end{proof}

\begin{proof}[Proof of Proposition \ref{prop_A_alpha_invertibility}]
The invertibility is proved following similar lines as in Theorem \ref{thm_cap_def_momentum_existence}. As in \eqref{eq_app_A_8}, for $\varphi\in P_{\perp}L^2(\partial D)$, we define
\begin{equation} \label{eq_app_B_1}
\begin{aligned}
u(\bm{x})&=-\frac{1}{4\pi}\int_{\partial D}d\sigma_{\bm{y}}\varphi(\bm{y})\frac{1}{|\bm{x}-\bm{y}|}  \\
&\quad -\frac{1}{4\pi}\int_{\partial D}d\sigma_{\bm{y}}\varphi(\bm{y})\sum_{n\geq 1} \big[\frac{1}{|\bm{x}-\bm{y}-n\bm{e}_1|}-\frac{1}{|n|}\big] \\
&\quad -\frac{1}{4\pi}\int_{\partial D}d\sigma_{\bm{y}}\varphi(\bm{y})\sum_{n\leq -1} \big[\frac{1}{|\bm{x}-\bm{y}-n\bm{e}_1|}-\frac{1}{|n|}\big] \\
&=-\frac{1}{4\pi}\int_{\partial D}d\sigma_{\bm{y}}\varphi(\bm{y})\frac{1}{|\bm{x}-\bm{y}|}  \\
&\quad -\frac{1}{4\pi}\int_{\partial D}d\sigma_{\bm{y}}\varphi(\bm{y})\int_{0}^{\infty}\big[J_0(|x_{\perp}-y_{\perp}|t)e^{(x_1-y_1) t}-1\big]\frac{e^{-t}}{1-e^{-t}}dt \\
&\quad -\frac{1}{4\pi}\int_{\partial D}d\sigma_{\bm{y}}\varphi(\bm{y})\int_{0}^{\infty}\big[J_0(|x_{\perp}-y_{\perp}|t)e^{-(x_1-y_1) t}-1\big]\frac{e^{-t}}{1-e^{-t}}dt \\
&=: u^{free}(\bm{x})+u^{rem,0}_{+}(\bm{x})+u^{rem,0}_{-}(\bm{x}) .
\end{aligned}
\end{equation}
We will prove the following claims:
\begin{itemize}
    \item[(a)] $u^{rem,0}_{\pm}(\bm{x})$ are smooth near $\partial D$;
    \item[(b)] $u(\bm{x})$ are harmonic in $\bm{x}\in Y\backslash\partial D$;
    \item[(c)] $u(\bm{x})$ is periodic, i.e., $u(1,x_{\perp})=u(0,x_{\perp})$;
    \item[(d)] Across the surface, $u$ is continuous and the normal derivative $\partial_{\nu}u$ admits a nonzero jump:
\begin{equation*}
u\big|_{\partial D^{-}}=u\big|_{\partial D^{+}}=\mathcal{S}_{D}^{\alpha}[\varphi],\quad \partial_{\nu}u\big|_{\partial D^{-}}-\partial_{\nu}u\big|_{\partial D^{+}}=\varphi .
\end{equation*}
    \item[(e)] As $|x_{\perp}|\to \infty$, $u(\bm{x})$ admits the following asymptotics:
    \begin{equation*}
    u(\bm{x})=o\big(|x_{\perp}|^{-\frac{1}{2}}\big),\quad (\nabla_{\perp}u)(\bm{x})=o\big(|x_{\perp}|^{-\frac{1}{2}}\big).
    \end{equation*}
\end{itemize}
By Claim (a), one sees that $\mathcal{A}$ is Fredholm (arguing similarly as in the proof of Theorem \ref{thm_cap_def_momentum_existence}), and therefore, $P_{\perp}\mathcal{A}P_{\perp}$ is also Fredholm. Hence, we only need to prove that $P_{\perp}\mathcal{A}P_{\perp}$ is injective. Suppose that $P_{\perp}\mathcal{A}\varphi=0$ with $\varphi\in P_{\perp}L^2(\partial D)$. Then, there exists $c\in\mathbb{R}$ such that
\begin{equation*}
\mathcal{A}\varphi=c\mathbbm{1}_{\partial D} .
\end{equation*}
Using Claims (b)-(e) and applying the Gauss-Green formula as in Theorem \ref{thm_cap_def_momentum_existence}, we obtain
\begin{equation*}
\int_{Y\backslash \overline{D}}|\nabla u|^2+\int_{D}|\nabla u|^2
=c\int_{\partial D} \partial_{\nu} u\big|_{\partial D^{+}}-\partial_{\nu} u\big|_{\partial D^{-}} \overset{(i)}{=}c\int_{\partial D} \varphi \overset{(ii)}{=}0,
\end{equation*}
where the identity (i) follows from the jump formula (d), and (ii) is a consequence of the assumption $\varphi\in P_{\perp}L^2(\partial D)$. Hence, we conclude that $u(\bm{x})$ is constant for $\bm{x}\in Y$. Then, by applying the jump formula (d) again, we conclude that $\varphi =0$.

Now, we prove our Claims (a)-(e), among which (c) follows directly from the definition. On the other hand, (d) is a consequence of (a) and the jump formula for the single-layer potential. We also note that estimate \eqref{eq_app_A_9} holds for $u^{rem,0}_{\pm}(\bm{x})$ and hence, it implies (a) and (b). For (e), we need to be careful: the proof in Proposition \ref{prop_cap_def_momentum_existence_proof_2} does not apply directly because the integrand $J_0(|x_{\perp}-y_{\perp}|t)e^{(x_1-y_1) t}-1$ in \eqref{eq_app_A_9} does not vanish as $|x_{\perp}|\to \infty$. Nevertheless, the assumption $\varphi\in P_{\perp}L^2(\partial D)$ provides convenience. For example,
\begin{equation*}
\begin{aligned}
\int_{\partial D}d\sigma_{\bm{y}}\varphi(\bm{y})\sum_{n\geq 1} \big[\frac{1}{|\bm{x}-\bm{y}-n\bm{e}_1|}-\frac{1}{|n|}\big]
&=\int_{\partial D}d\sigma_{\bm{y}}\varphi(\bm{y})\sum_{n\geq 1} \big[\frac{1}{|\bm{x}-\bm{y}-n\bm{e}_1|}-\frac{1}{|\bm{x}-\bm{y}_0-n\bm{e}_1|}\big]
\end{aligned}
\end{equation*}
for any fixed $\bm{y}_0\in \partial D$. By the mean-value theorem, one sees that
\begin{equation*}
\begin{aligned}
\Big|\frac{1}{|\bm{x}-\bm{y}-n\bm{e}_1|}-\frac{1}{|\bm{x}-\bm{y}_0-n\bm{e}_1|}\Big|
&\lesssim |\bm{y}-\bm{y}_0|\sup_{t\in [0,1]}\frac{1}{|\bm{x}-t\bm{y}-(1-t)\bm{y}_0-n\bm{e}_1|^2} \\
&\lesssim \frac{1}{|x_{\perp}|^2+n^2},
\end{aligned}
\end{equation*}
where the last inequality follows from the fact that $\bm{y},\bm{y}_0\in\partial D$. Hence, we obtain
\begin{equation*}
|u^{rem,0}_{\pm}(\bm{x})|\lesssim \sum_{n\geq 1}\frac{1}{|x_{\perp}|^2+n^2} =\mathcal{O}(|x_{\perp}|^{-1}),
\end{equation*}
as $|x_{\perp}|\to\infty$. The estimate of spatial derivatives is similar, for which we skip the details and directly show the result:
\begin{equation*}
\nabla_{x_{\perp}} u^{rem,0}_{\pm} =\mathcal{O}(|x_{\perp}|^{-2}) .
\end{equation*}
These estimates, together with the estimate of the free part $u^{free}$ using the asymptotics of the free-space Green function, conclude the proof of claim (e).
\end{proof}

\section{Proofs of Propositions \ref{lem_real_even_cap_momentum} and \ref{prop_off_diag_decay_proof_2}}\label{sec:app_C}
\begin{proof}[Proof of Proposition \ref{lem_real_even_cap_momentum}]
We first prove $\widehat{\mathcal{C}}^{\alpha}\in \mathbb{R}$. In fact, by definition of the quasi-periodic single-layer potential operator $\mathcal{S}_{D}^{\alpha}$, it is direct to check its Hermiticity, i.e., $(\mathcal{S}_{D}^{\alpha}[u],u)_{L^2(\partial D)}\in\mathbb{R}$ for any $u\in L^2(\partial D)$. This, together with the invertibility proved in Theorem \ref{thm_cap_def_momentum_existence}, implies that $(\mathcal{S}_{D}^{\alpha})^{-1}$ is also Hermitian. Hence $\widehat{\mathcal{C}}^{\alpha}=\big((\mathcal{S}_{D}^{\alpha})^{-1}[\mathbbm{1}_{\partial D}],\mathbbm{1}_{\partial D}\big)_{L^2(\partial D)}\in \mathbb{R}$.

On the other hand, the definition of $\mathcal{S}_{D}^{\alpha}$ also implies that $\mathcal{S}_{D}^{-\alpha}=\mathcal{T}\mathcal{S}_{D}^{\alpha}$, with $\mathcal{T}$ being the complex conjugation operator, using which one directly checks $(\mathcal{S}_{D}^{-\alpha})^{-1}[\mathbbm{1}_{\partial D}]=\mathcal{T}(\mathcal{S}_{D}^{\alpha})^{-1}[\mathbbm{1}_{\partial D}]$. This indicates that $\widehat{\mathcal{C}}^{\alpha}=\overline{\widehat{\mathcal{C}}^{-\alpha}}$, and completes the proof together with the fact that $\widehat{\mathcal{C}}^{\alpha}\in \mathbb{R}$.
\end{proof}

\begin{proof}[Proof of Proposition \ref{prop_off_diag_decay_proof_2}]
We first estimate $I_{1,near}^{(1)}(n)$. Using the results of Theorem \ref{thm_cap_def_momentum_estimate}, it follows that
\begin{equation} \label{eq_app_C_1}
|I_{1,near}^{(1)}(n)|\lesssim\int_{0}^{n^{-1}}\frac{1}{\alpha\log^3(1/\alpha)}|\sin (n\alpha)|d\alpha
\leq \frac{1}{\log^3 n}\int_{0}^{1}\frac{|\sin (t)|}{t}dt
=\mathcal{O}(\frac{1}{\log^3 n}).
\end{equation}
The estimate of $I_{1,far}^{(1)}(n)$ is more delicate. For $p\in(0,1)$, we decompose the integral as
\begin{equation} \label{eq_app_C_2}
I_{1,far}^{(1)}(n)=\int_{n^{-1}}^{n^{-p}}(\partial_{\alpha}r(\alpha))\sin (n\alpha)d\alpha + \int_{n^{-p}}^{\alpha_0/2}(\partial_{\alpha}r(\alpha))\sin (n\alpha)d\alpha.
\end{equation}
For the first part, we proceed similarly as in \eqref{eq_app_C_1}:
\begin{equation}
\big|\int_{n^{-1}}^{n^{-p}}(\partial_{\alpha}r(\alpha))\sin (n\alpha)d\alpha\big|
\lesssim \int_{n^{-1}}^{n^{-p}}\frac{1}{\alpha\log^3(1/\alpha)}d\alpha
=\frac{1-p^2}{2p^2}\frac{1}{\log^2 n} .
\end{equation}
For the other part, an integration by parts shows that
\begin{equation}  \label{eq_app_C_3}
\begin{aligned}
\int_{n^{-p}}^{\alpha_0/2}(\partial_{\alpha}r(\alpha))\sin (n\alpha)d\alpha
&=-\frac{1}{n}(\partial_{\alpha}r(\alpha))\cos (n\alpha)\big|_{\alpha=\frac{1}{n^p}}^{\alpha_0/2}+\frac{1}{n}\int_{n^{-p}}^{\alpha_0/2}(\partial_{\alpha}^2 r(\alpha))\cos (n\alpha)d\alpha \\
&\overset{(i)}{\lesssim}\frac{1}{n}\frac{n^p}{\log^3 n}+\frac{1}{n}\int_{n^{-p}}^{\alpha_0/2}\frac{1}{\alpha^2 \log^3(1/\alpha)}d\alpha \\
&\overset{(ii)}{\lesssim} \frac{1}{n^{1-p}\log^3 n} +\frac{1}{n^{1-p}}\int_{n^{-p}}^{\alpha_0/2}\frac{1}{\alpha \log^3(1/\alpha)}d\alpha \\
&\lesssim \frac{1}{n^{1-p}},
\end{aligned}
\end{equation}
where the inequalities (i) and (ii) follow from Theorem \ref{thm_cap_def_momentum_estimate}. The proof of Proposition \ref{prop_off_diag_decay_proof_2} is complete by \eqref{eq_app_C_1}-\eqref{eq_app_C_3}.
\end{proof}

\section{Proof of Lemma \ref{lem_inverse_closed}}\label{sec:app_D}

We only need to prove \eqref{eq_inverse_closed_1} as the estimate \eqref{eq_inverse_closed_2} follows directly from the Riesz projection formula. The proof of \eqref{eq_inverse_closed_1} is based on the generalised Baskakov inverse-closedness theorem (see \cite[Theorem 4.2]{grochenig2010convergence}) stated as follows: let $g:\mathbb{Z}\to \mathbb{R}^{+}$ be a non-negative function satisfying 
\begin{itemize}
    \item[(i)] (normalised and even) $g(0)=1$ and $g(-n)=n$;
    \item[(ii)] (sub-multiplicative) $g(n+m)\leq g(n)g(m)$ for any $n,m\in\mathbb{Z}$;
    \item[(iii)] (sub-convolutive) $g^{-1}\ast g^{-1}\leq Cg^{-1}\in \ell^1(\mathbb{Z})$ for some $C>0$, 
\end{itemize}
and let $\mathcal{E}_g\subset \mathcal{B}(\ell^2(\mathbb{Z}))$ be the collection of operators
\begin{equation*}
\mathcal{E}_g:=\big\{ A\in \mathcal{B}(\ell^2(\mathbb{Z})):\, \sup_{n,m\in \mathbb{Z}}|A(n,m)|g(n-m)<\infty \big\}.
\end{equation*}
Then the set $\mathcal{E}_g$ is a Banach algebra and is inverse-closed in the sense that
\begin{equation*}
A\in \mathcal{E}_g,\, A^{-1}\in \mathcal{B}(\ell^2(\mathbb{Z}))\Longrightarrow A^{-1}\in \mathcal{E}_g.
\end{equation*}
For our application, we take
\begin{equation*}
g(n):=(1+|n|)\log^2(e+|n|),
\end{equation*}
which obviously satisfies condition (i). Hence, after verifying (ii)-(iii), and noticing that $\tilde{\mathcal{C}}\in \mathcal{E}_{g}$ thanks to Theorem \ref{thm_off_diag_decay}, we conclude the proof of \eqref{eq_inverse_closed_1}. In the sequel, we check conditions (ii) and (iii) by elementary calculations, respectively.

{\color{blue}Proof of (ii):} First, it is direct to see that
\begin{equation} \label{eq_inverse_closed_proof_1}
1+|n+m|\leq (1+|n|)(1+|m|).
\end{equation}
On the other hand, we claim
\begin{equation} \label{eq_inverse_closed_proof_2}
\log(e+|n+m|)\leq \log(e+|n|)\log(e+|m|) .
\end{equation}
Then (ii) follows directly from \eqref{eq_inverse_closed_proof_1} and \eqref{eq_inverse_closed_proof_2}. To show \eqref{eq_inverse_closed_proof_2}, we note that
\begin{equation*}
e+|n+m|\leq e+|n|+|m|\leq \frac{(e+|n|)(e+|m|)}{e}.
\end{equation*}
Hence,
\begin{equation*}
\log(e+|n+m|)\leq \log(e+|n|)+\log(e+|m|)-1 .
\end{equation*}
Then \eqref{eq_inverse_closed_proof_2} is proved by noticing that
\begin{equation*}
\log(e+|n|)+\log(e+|m|)\leq 1+\log(e+|n|)\log(e+|m|)
\end{equation*}
since $\log(e+|n|),\log(e+|m|)\geq 1$.

{\color{blue}Proof of (iii):} The estimate $g^{-1}\in \ell^1(\mathbb{Z})$ is clear. To estimate the convolution $g^{-1}\ast g^{-1}$, we note that
\begin{equation} \label{eq_inverse_closed_proof_3}
\begin{aligned}
(g^{-1}\ast g^{-1})(n)&=\sum_{m\in\mathbb{Z} }g^{-1}(n-m)g^{-1}(m) \\
&=\sum_{m:\, |m-n|\leq \frac{|n|}{2}}g^{-1}(n-m)g^{-1}(m)+ \sum_{m:\, |m-n|> \frac{|n|}{2}}g^{-1}(n-m)g^{-1}(m) .
\end{aligned}
\end{equation}
For the first sum, observe that
\begin{equation*} 
g^{-1}(m)\leq g^{-1}(\frac{n}{2}) \lesssim g^{-1}(n)
\end{equation*}
since $|m|\geq |n|-\frac{|n|}{2}=\frac{|n|}{2}$, $g^{-1}$ is decreasing, and $\sup_{n\in\mathbb{Z}}g(\frac{n}{2})/g(n)<\infty$. Thus,
\begin{equation}  \label{eq_inverse_closed_proof_4}
\sum_{m:\, |m-n|\leq \frac{|n|}{2}}g^{-1}(n-m)g^{-1}(m)
\lesssim g^{-1}(n)\sum_{m:\, |m-n|\leq \frac{|n|}{2}}g^{-1}(n-m)
\leq \|g^{-1}\|_{\ell^1(\mathbb{Z})}g^{-1}(n) .
\end{equation}
Similarly, we have
\begin{equation}  \label{eq_inverse_closed_proof_5}
\sum_{m:\, |m-n|> \frac{|n|}{2}}g^{-1}(n-m)g^{-1}(m)
\leq g^{-1}(\frac{n}{2})\sum_{m:\, |m-n|> \frac{|n|}{2}}g^{-1}(m)
\lesssim \|g^{-1}\|_{\ell^1(\mathbb{Z})}g^{-1}(n) .
\end{equation}
Hence, the proof of (iii) is complete by \eqref{eq_inverse_closed_proof_3}-\eqref{eq_inverse_closed_proof_5}.

\printbibliography

\end{document}